\documentclass[submission,copyright,creativecommons]{eptcs}
\providecommand{\event}{QPL 2024}
\pdfoutput=1

\usepackage{amsthm, amsmath, amssymb, dsfont, physics, textcomp}
\usepackage{tikz}
\usepackage{tikzit}

\tikzstyle{dot}=[inner sep=0.3mm, minimum width=2mm, minimum height=2mm, draw, shape=circle, font={\footnotesize}, tikzit fill=magenta]
\tikzstyle{white dot}=[dot, fill=white, text depth=-0.2mm, tikzit category=ZH-pf]
\tikzstyle{gray dot}=[dot, fill=black, text depth=-0.2mm, tikzit category=ZH-pf]
\tikzstyle{gray phase dot}=[gray dot, fill=black, tikzit fill=magenta]
\tikzstyle{hadamard}=[fill=yellow, draw, inner sep=0.6mm, minimum height=1.5mm, minimum width=1.5mm, shape=rectangle, tikzit shape=rectangle, tikzit category=ZH-pf]
\tikzstyle{small hadamard}=[fill=yellow, draw, inner sep=0.6mm, minimum height=1.5mm, minimum width=1.5mm, tikzit shape=rectangle]
\tikzstyle{small gray hadamard}=[fill=black, draw, inner sep=0.6mm, minimum height=1.5mm, minimum width=1.5mm, tikzit shape=rectangle]
\tikzstyle{halfscalar}=[star, fill=black, draw=black, minimum size=6pt, inner sep=0pt]
\tikzstyle{scalar}=[regular polygon, fill=black, draw=black, minimum size=6pt, inner sep=0pt, regular polygon sides=3]
\tikzstyle{box}=[shape=rectangle, text height=1.5ex, text depth=0.25ex, yshift=0.2mm, fill=white, draw=black, minimum height=3mm, minimum width=5mm, font={\small}]
\tikzstyle{smallbox}=[draw, fill=white, inner sep=0.6mm, minimum height=1.5mm, minimum width=1.5mm, font={\scriptsize}, shape=rectangle]
\tikzstyle{multiply box}=[fill=white, draw, inner sep=0.6mm, minimum height=1ex, minimum width=2ex, text depth=0ex, text height=1.25ex, tikzit shape=rectangle, font={$\cdot$}]
\tikzstyle{Z dot}=[inner sep=0mm, minimum size=2mm, shape=circle, draw=black, fill={rgb,255: red,216; green,248; blue,216}, tikzit fill={rgb,255: red,216; green,248; blue,216}]
\tikzstyle{Z phase dot}=[minimum size=5mm, font={\footnotesize\boldmath}, shape=rectangle, rounded corners=2mm, inner sep=0.2mm, outer sep=-2mm, scale=0.8, tikzit shape=circle, draw=black, fill={zx_green}, tikzit fill={rgb,255: red,216; green,248; blue,216}, tikzit draw=blue]
\tikzstyle{X dot}=[Z dot, shape=circle, draw=black, fill={zx_red}]
\tikzstyle{X phase dot}=[Z phase dot, tikzit shape=circle, tikzit draw=blue, fill={zx_red}, font={\footnotesize\color{black}\boldmath}]
\tikzstyle{lbl}=[font={\scriptsize}]
\tikzstyle{monoid}=[shape=semicircle, fill=white, draw=black, inner sep=0.5mm, rotate=180]
\tikzstyle{gray monoid}=[shape=semicircle, fill=black, draw=black, inner sep=0.5mm, rotate=180]
\tikzstyle{red monoid}=[shape=semicircle, fill={zx_red}, draw=black, inner sep=0.5mm, rotate=180]
\tikzstyle{yellow monoid}=[shape=semicircle, fill=yellow, draw=black, inner sep=0.5mm, rotate=180]
\tikzstyle{comonoid}=[shape=semicircle, fill=white, draw=black, inner sep=0.5mm]
\tikzstyle{gray comonoid}=[shape=semicircle, fill=black, draw=black, inner sep=0.5mm]
\tikzstyle{red comonoid}=[shape=semicircle, fill={zx_red}, draw=black, inner sep=0.5mm]
\tikzstyle{yellow comonoid}=[shape=semicircle, fill=yellow, draw=black, inner sep=0.5mm]
\tikzstyle{green bra}=[fill={rgb,255: red,216; green,248; blue,216}, draw=black, regular polygon, regular polygon sides=3, tikzit fill={rgb,255: red,216; green,248; blue,216}, tikzit draw=black, inner sep=0 mm, outer sep=0 mm, shape border rotate=0, font={\tiny}, minimum height=9mm]
\tikzstyle{green ket}=[fill={rgb,255: red,216; green,248; blue,216}, draw=black, regular polygon, regular polygon sides=3, tikzit fill={rgb,255: red,216; green,248; blue,216}, tikzit draw=black, inner sep=0 mm, outer sep=0 mm, shape border rotate=180, font={\tiny}, minimum height=9mm]
\tikzstyle{ket}=[draw=black, regular polygon, regular polygon sides=3, tikzit draw=black, inner sep=0 mm, outer sep=0 mm, shape border rotate=180, font={\tiny}, minimum height=9mm]
\tikzstyle{gn}=[inner sep=0mm, minimum size=2mm, shape=circle, draw=black, fill={rgb,255: red,216; green,248; blue,216}, tikzit fill={rgb,255: red,216; green,248; blue,216}]
\tikzstyle{srn}=[Z dot, shape=circle, draw=black, fill={rgb,255: red,232; green,165; blue,165}, tikzit fill={rgb,255: red,232; green,165; blue,165}, minimum size=1mm]
\tikzstyle{rn}=[Z dot, shape=circle, draw=black, fill={rgb,255: red,232; green,165; blue,165}, tikzit fill={rgb,255: red,232; green,165; blue,165}]
\tikzstyle{pn}=[Z dot, fill={rgb,255: red,255; green,200; blue,240}, tikzit fill={rgb,255: red,255; green,200; blue,240}]
\tikzstyle{yn}=[Z dot, fill=yellow, tikzit fill=yellow]
\tikzstyle{H box}=[rectangle, fill=yellow, draw=black, xscale=1, yscale=1, font={\small}, inner sep=0.75pt, minimum width=0.15cm, minimum height=0.15cm, tikzit shape=rectangle]
\tikzstyle{yellow hadamard}=[fill=yellow, draw=black, shape=rectangle, inner sep=0.6mm, minimum height=1.5mm, minimum width=1.5mm]
\tikzstyle{ug}=[regular polygon, regular polygon sides=3, fill={zx_red}, draw=black, inner sep=0pt, minimum width=1em, tikzit draw=blue]
\tikzstyle{st}=[star, star points=5, fill=white, draw=black, inner sep=1.2pt, line width=1.2pt, tikzit fill=blue, tikzit draw=red, tikzit category=ZH-pf]
\tikzstyle{not}=[fill={rgb,255: red,180; green,180; blue,180}, draw=black, shape=circle, font={$\neg$}, dot]
\tikzstyle{bbindex}=[font={\color{blue}\footnotesize}]
\tikzstyle{wide point}=[fill=white, draw, shape=isosceles triangle, shape border rotate=-90, isosceles triangle stretches=true, inner sep=0pt, minimum width=1.5cm, minimum height=6.12mm, yshift=-0.0mm]
\tikzstyle{medium gray box}=[semilarge box, fill={rgb,255: red,180; green,180; blue,180}]
\tikzstyle{small box}=[rectangle, inline text, fill=white, draw, minimum height=5mm, yshift=-0.5mm, minimum width=5mm, font={\small}]
\tikzstyle{small gray box}=[small box, fill={rgb,255: red,180; green,180; blue,180}]
\tikzstyle{medium box}=[rectangle, inline text, fill=white, draw, minimum height=5mm, yshift=-0.5mm, minimum width=8mm, font={\small}]
\tikzstyle{wire label}=[font={\footnotesize}, tikzit fill=blue, anchor=west, shape=rectangle, inner sep=1pt, xshift=-1mm]
\tikzstyle{0 control}=[minimum size=1mm, shape=circle, draw=black, fill=white, font={\footnotesize\boldmath}, inner sep=0.5pt]
\tikzstyle{Xplus dot}=[Z dot, shape=circle, draw=black, fill=red]
\tikzstyle{neg0 control}=[0 control, inner sep=0 mm, fill=white, draw=black]
\tikzstyle{hz}=[small hadamard, fill={rgb,255: red,216; green,248; blue,216}, shape=rectangle, tikzit fill={rgb,255: red,216; green,248; blue,216}, minimum height=1.5 mm, minimum width=0.75 mm, tikzit draw=black, text width=0.1 mm, inner sep=0.1 mm]
\tikzstyle{hx}=[small hadamard, fill=red, shape=rectangle, tikzit fill=red, minimum height=1.5 mm, minimum width=0.75 mm, tikzit draw=black, text width=0.1 mm, inner sep=0.1 mm]
\tikzstyle{ctrl}=[inner sep=0mm, minimum size=1mm, shape=circle, draw=black, fill=black]
\tikzstyle{targ}=[draw=black, fill=white, font={\footnotesize\boldmath}, minimum size=0.5mm, inner sep=-0.5mm, shape=circle, tikzit shape=circle, tikzit draw=black]
\tikzstyle{XD dot}=[shape=XDdot, inner sep=2pt, draw=black, tikzit fill={rgb,255: red,255; green,191; blue,191}]
\tikzstyle{XD phase dot}=[shape=XDdotphase, minimum size=4.75mm, font={\footnotesize}, inner sep=.1mm, outer sep=0mm, scale=0.8, tikzit shape=circle, rounded corners=1.9mm, draw=black, tikzit fill={rgb,255: red,255; green,191; blue,191}, tikzit draw=blue]

\tikzstyle{gray}=[-, draw={blue!60!white}, tikzit draw=blue]
\tikzstyle{blue}=[-, draw={blue!60!white}, tikzit draw=blue]
\tikzstyle{brace edge}=[-, tikzit draw=blue, decorate, decoration={brace,amplitude=1mm,raise=-1mm}]
\tikzstyle{diredge}=[->]
\tikzstyle{not edge}=[-, dashed, dash pattern=on 2pt off 1.5pt, thick, draw={rgb,255: red,255; green,68; blue,68}]
\tikzstyle{thin}=[-, line width=0.10mm]

\usetikzlibrary{calc,decorations.pathreplacing,positioning,shapes,cd}
\usetikzlibrary{positioning}
\usetikzlibrary{shapes}
\usetikzlibrary{chains}
\usepackage{quantikz}
\usepackage{xurl}
\usepackage{hyperref}
\usepackage[capitalize,nameinlink,sort,noabbrev]{cleveref}

\newcommand{\Z}{\mathbb{Z}}
\newcommand{\N}{\mathbb{N}}
\newcommand{\C}{\mathbb{C}}

\newcommand{\Zz}{\mathbb{Z}[\omega]}
\newcommand{\Zzk}[1]{\mathbb{Z}[\omega_{#1}]}
\newcommand{\Tz}{\mathbb{Z}[1/3,\omega]}

\newcommand{\Tzk}[1]{\mathbb{Z}[1/3,\omega_{#1}]}
\newcommand{\s}[1]{\{#1\}}
\newcommand{\un}{\mathrm{U}_n(\Tz)}
\newcommand{\unn}[1]{\mathrm{U}_{#1}(\Tz)}
\newcommand{\uu}{\mathrm{U}(\Tz)}
\newcommand{\unk}[1]{\mathrm{U}_n(\Tzk{#1})}
\newcommand{\unkk}[2]{\mathrm{U}_{#1}(\Tzk{#2})}
\newcommand{\uuk}[1]{\mathrm{U}(\Tzk{#1})}
\newcommand{\gens}{\mathcal{S}}
\newcommand{\gensn}{\mathcal{S}_n}
\newcommand{\gensnn}[1]{\mathcal{S}_{#1}}
\newcommand{\lde}{\operatorname{lde}} 
\newcommand{\calu}{\mathcal{U}}
\newcommand{\calv}{\mathcal{V}}
\newcommand{\calw}{\mathcal{W}}
\newcommand{\vecc}{\ket{c}}

\newcommand{\rotk}[1]{T_{#1}}
\newcommand{\gatesetk}[1]{\mathcal{G}_{#1}}

\newcommand{\gatesetkone}{\mathcal{G}}

\newtheoremstyle{break}
  {}
  {}
  {\itshape}
  {}
  {\bfseries}
  {.}
  {\newline}
  {}

\theoremstyle{plain}
\newtheorem{theorem}{Theorem}[section]
\newtheorem{lemma}[theorem]{Lemma}
\newtheorem{proposition}[theorem]{Proposition}
\newtheorem*{proposition*}{Proposition}
\newtheorem{corollary}[theorem]{Corollary}
\theoremstyle{break}

\theoremstyle{definition}
\newtheorem{definition}[theorem]{Definition}

\theoremstyle{remark}
\newtheorem{remark}[theorem]{Remark}

\newcommand{\urlaltve}[2]{\href{#2}{\nolinkurl{#1}}}

\usepackage{xcolor}
 \hypersetup{
     colorlinks,
     linkcolor={red!80!black},
     citecolor={green!80!black},
     urlcolor={blue!80!black}
 }

\title{
  Exact Synthesis of Multiqutrit Clifford-Cyclotomic Circuits
  }

\author{Andrew N. Glaudell
\institute{Photonic Inc.}
\email{andrewglaudell@gmail.com} \and
Neil J. Ross
\institute{Dalhousie University}
\email{neil.jr.ross@dal.ca} \and
John van de Wetering
\institute{University of Amsterdam}
\email{john@vdwetering.name} \and 
Lia Yeh
\institute{University of Oxford}
\email{lia.yeh@cs.ox.ac.uk}
}

\def\titlerunning{Exact Synthesis of Multiqutrit Clifford-Cyclotomic
  Circuits}

\def\authorrunning{A.~N.~Glaudell, N.~J.~Ross, J. van de Wetering \&
  L.~Yeh}

\date{\today}

\begin{document}

\maketitle

\begin{abstract}
  It is known that the matrices that can be exactly represented by a
  multiqubit circuit over the Toffoli+Hadamard, Clifford+$T$, or, more
  generally, Clifford-cyclotomic gate set are precisely the unitary
  matrices with entries in the ring $\Z[1/2,\zeta_k]$, where $k$ is a
  positive integer that depends on the gate set and $\zeta_k$ is a
  primitive $2^k$-th root of unity. In the present paper, we establish
  an analogous correspondence for qutrits. We define the multiqutrit
  Clifford-cyclotomic gate set of degree $3^k$ by extending the
  classical qutrit gates $X$, $CX$, and $CCX$ with the Hadamard gate
  $H$ and the $T_k$ gate $T_k=\mathrm{diag}(1,\omega_k, \omega_k^2)$,
  where $\omega_k$ is a primitive $3^k$-th root of unity. This gate
  set is equivalent to the qutrit Toffoli+Hadamard gate set when
  $k=1$, and to the qutrit Clifford+$T_k$ gate set when $k>1$. We then
  prove that a $3^n\times 3^n$ unitary matrix $U$ can be represented
  by an $n$-qutrit circuit over the Clifford-cyclotomic gate set of
  degree $3^k$ if and only if the entries of $U$ lie in the ring
  $\mathbb{Z}[1/3,\omega_k]$.
\end{abstract}

\section{Introduction}
\label{sec:intro}

\subsection{Background}
\label{ssec:background}

In quantum computing, \textbf{synthesis} refers to the process of
converting a representation of a unitary into a quantum circuit. In
\textbf{exact synthesis} the unitary is typically given as a matrix,
and the goal is to produce a circuit that implements the matrix
exactly. This is in contrast to \textbf{approximate synthesis}, where
the circuit is only required to implement the given matrix up to some
prescribed error budget.

A solution to an exact synthesis problem for a gate set $\mathcal{G}$
sometimes characterizes the unitary matrices that can be exactly
represented by a circuit over $\mathcal{G}$. For instance, the 
matrices with entries in the ring $\mathbb{Z}[1/2]$ of dyadic
rationals corresponds precisely to the unitary matrices that can be
represented using the Toffoli gate and the tensor product $H\otimes H$
of the Hadamard gate with itself \cite{Amy2020}. Similarly,
Clifford+$T$ circuits correspond to unitary matrices with entries in
$\mathbb{Z}[1/2, e^{2\pi i /8}]$ \cite{Giles2013a}. More generally, it
was recently shown that multiqubit circuits over the
Clifford-cyclotomic gate set of degree $k$, which extends the Clifford 
gate set with a $z$-rotation by angle $2\pi/2^k$, correspond to
unitary matrices with entries the ring $\mathbb{Z}[1/2,e^{2\pi
    i/2^k}]$ \cite{qubitcyclo}.

In this paper, we consider the exact synthesis problem for
qutrits. Like for qubits, fault-tolerant universal quantum computation
has been theoretically devised for qutrits through magic state
distillation~\cite{AnwarH2012qutritmsd,
  CampbellE2012tgatedistillation, PrakashS2020qutritgolay} or gauge
fixing of colour codes~\cite{WatsonF2015quditcolorcodes}. In recent
years, qudit operations have been demonstrated on many experimental
platforms~\cite{HrmoP2023quditions,KasperV2022quditcoldatoms,WeggemansJ2022quditRydberg,YurtalanM2020Walsh-Hadamard},
with error rates competitive to qubit
operations~\cite{RingbauerM2018ququartphotonic,
  ChiY2022quditphotonicprocessor}. Qutrit exact synthesis problems,
however, have received less attention than their qubit counterparts
and only a few results exist: a normal form for single-qutrit
Clifford+$T$ unitaries \cite{qutritGlaudell,prakashqutrit}, a proof
that all classically reversible functions on trits can be implemented
using Clifford+$T$ circuits \cite{liaRC}, and an exact synthesis
result for single-qutrit Clifford+$R$ unitaries
\cite{KalraAR2024synth1qutrit}.

Let $k$ be a positive integer and let $\omega_k\in\C$ be the
\textbf{primitive $3^k$-th root of unity} $\omega_k = e^{2\pi
  i/3^k}$. For simplicity, we write $\omega$ for $\omega_1$. The
single-qutrit \textbf{Pauli $X$} gate, \textbf{Pauli $Z$} gate,
\textbf{phase} gate $S$, and \textbf{Hadamard} gate $H$ are defined
below.
\[
X = 
\begin{bmatrix}
\cdot     & \cdot    & 1 \\
1 & \cdot & \cdot \\
\cdot & 1    & \cdot \\
\end{bmatrix}
\qquad
Z = 
\begin{bmatrix}
1     & \cdot    & \cdot \\
\cdot & \omega & \cdot \\
\cdot & \cdot    & \omega^2 \\
\end{bmatrix}
\qquad
S = 
\begin{bmatrix}
1     & \cdot    & \cdot \\
\cdot & 1 & \cdot \\
\cdot & \cdot    & \omega \\
\end{bmatrix}
\qquad
H = \frac{-\omega^2}{\sqrt{-3}}
\begin{bmatrix}
1 & 1        & 1 \\
1 & \omega   & \omega^2 \\
1 & \omega^2 & \omega \\
\end{bmatrix}
\]
The two-qutrit \textbf{controlled-$X$} gate $CX$ is the permutation
matrix whose action on the computational basis is defined by
$\ket{i}\ket{j} \mapsto \ket{i}\ket{i+j}$, with addition performed
modulo 3. The three-qutrit \textbf{doubly-controlled-$X$} gate $CCX$
(or \textbf{Toffoli} gate) is similarly defined by
$\ket{i}\ket{j}\ket{k} \mapsto \ket{i}\ket{j}\ket{k+ij}$. The gate set
$\s{H, S, CX}$ is the \textbf{Clifford} gate set. Now define the
single-qutrit $T_k$ gate
\[
T_k = 
\begin{bmatrix}
1     & \cdot    & \cdot \\
\cdot & \omega_k & \cdot \\
\cdot & \cdot    & \omega^2_k \\
\end{bmatrix}.
\]
When $k=2$, $T_k$ is the qutrit $T$ gate \cite{HowardM2012quditTgate}.

The \textbf{Clifford-cyclotomic gate set of degree $3^k$} is the gate
set $\gatesetk{k} = \s{X, CX, CCX, H, \rotk{k}}$. When $k=1$, we have
$\rotk{1} = Z = H X H^\dagger$, so that the Clifford-cyclotomic gate
set of degree $3$ is equivalent to the qutrit
\textbf{Toffoli+Hadamard} gate set \cite{qupitZH}. As we will show
below, when $k\geq 2$, the gate set $\gatesetk{k}$ is equivalent (up
to a single ancillary qutrit) to the \textbf{Clifford+$\rotk{k}$} gate
set $\s{H, S, CX, \rotk{k}}$. In particular, the Clifford-cyclotomic
gate set of degree 9 is equivalent to the well-known qutrit
\textbf{Clifford+$T$} gate set
\cite{qutritGlaudell,metaplectic,prakashqutrit,liaRC}. Because
$T_{k+1}^3 = T_k$, the Clifford-cyclotomic gate sets form a hierarchy
of universal gate sets whose first level is given by the
Toffoli+Hadamard gate set, whose second level is given by the
Clifford+$T$ gate set, and whose subsequent levels are given by finer
and finer extensions of the Clifford gate set.

Now consider the ring $\Tzk{k}$, which can be defined as the smallest
unital subring of $\C$ containing $1/3$ and $\omega_k$. Since
$-\omega^2/\sqrt{-3} = \omega^2 (1-\omega)/3$, the entries of $X$,
$CX$, $CCX$, $H$, and $T_k$ lie in $\Tzk{k}$. Hence, any $n$-qutrit
circuit over $\gatesetk{k}$ must represent a unitary matrix with
entries in $\Tzk{k}$. The purpose of this paper is to show that the
converse implication is also true.

\subsection{Contributions}
\label{ssec:contribs}

We show that a $3^n\times 3^n$ unitary matrix $U$ can be exactly
represented by an $n$-qutrit circuit over the Clifford-cyclotomic gate
set of degree $3^k$ if and only if the entries of $U$ belong to the
ring $\Tzk{k}$. Furthermore, we show that no more than $k+1$ ancillae
are required for this purpose.

We therefore solve the exact synthesis problem for multiqutrit
Toffoli+Hadamard circuits, multiqutrit Clifford+$T$ circuits, and,
more generally, multiqutrit Clifford-cyclotomic circuits. To the best
of our knowledge, this is the first time that a multiqudit exact
synthesis result is established for any prime $d>2$.

A similar hierarchy of Clifford-cyclotomic gate sets exists for
qubits, and the correspondence between Clifford-cyclotomic circuits
and matrices with entries in rings of algebraic integers also holds in
that case \cite{qubitcyclo}. Following \cite{qubitcyclo}, we prove our
result inductively. We first show that circuits over $\gatesetk{1}$
correspond to unitary matrices over $\Tz$ by reasoning as in
\cite{Amy2020,Giles2013a,greylyn}. This serves as the base case of our
induction. Then, we use properties of the ring extension $\Tzk{k}
\subseteq \Tzk{k+1}$ and the theory of catalytic embeddings
\cite{catemb} to establish the inductive step.

\subsection{Contents}
\label{ssec:contents}

The paper is organized as follows. We discuss the necessary
number-theoretic prerequisites in \cref{sec:ringsandgroups}. In
\cref{sec:generators}, we introduce a convenient generating set for
the group $\un$ of $n$-dimensional unitary matrices with entries in
the ring $\Tz$, and in \cref{sec:circuits} we show that the elements
of this generating set can be represented by Clifford-cyclotomic
circuits of degree $3$ (explicit circuit decompositions are given in
\cref{app:decompositions}). We introduce catalytic embeddings in
\cref{sec:catembs}. We leverage the results of the previous sections
in \cref{sec:exactsynth} to prove our main result. We comment on the
complexity of the produced circuits in \cref{sec:complexity} and we
conclude in \cref{sec:conc}.

\paragraph{Disclaimer:}
After the present work was completed, it was brought to our attention
that related results were independently established in
\cite{KalraAR2024synthmultiqutrit}.

\section{Rings and Groups}
\label{sec:ringsandgroups}

In this section, we discuss the rings and groups which will be
important in the rest of the paper. In what follows, when $u$, $u'$,
and $v$ are elements of a ring $R$, we write $u\equiv_v u'$ if $u$ is
congruent to $u'$ modulo $v$, i.e., if $u-u' = rv$ for some $r\in R$.

\subsection{The Ring \texorpdfstring{$\Zzk{k}$}{Z[ωk]}}
\label{ssec:zzk}

\begin{definition}
  \label{def:zeta}
  Let $k\geq 1$. The \textbf{primitive $3^k$-th root of unity}
  $\omega_k\in\C$ is defined as $\omega_k=e^{2\pi i/3^k}$.
\end{definition}

We have, for $k>1$, $\omega_k^3 = \omega_{k-1}$, $\omega_k^{3^k}=1$,
$\omega_k^\dagger = \omega_k^{3^k-1}$, and $\omega_k^0 + \omega_k^1 +
\ldots +\omega_k^{3^k-1} =0$. As mentioned in \cref{sec:intro}, we
often write $\omega$ for $\omega_1$.

\begin{definition}
  \label{def:cyclotk}
  Let $k\geq 1$. The ring $\Zzk{k}$ of \textbf{cyclotomic integers of
    degree $3^k$} is the smallest subring of $\C$ that contains
  $\omega_k$.
\end{definition}

The ring $\Zzk{k}$ can be defined in a variety of ways
\cite{cyclo}. It will be useful for our purposes to note that
$\Zz=\s{a+b\omega \mid a,b\in\Z}$, and that, for $k\geq 2$,
\[
\Zzk{k} = \s{a+b\omega_k + c\omega_k^2 \mid a,b,c\in\Zzk{k-1}}.
\]
Furthermore, the expression of an element of $\Zzk{k}$ as a linear
combination of elements of $\Zzk{k-1}$ is unique. The ring $\Zzk{k}$
is closed under complex conjugation and, for $k\geq 2$, we have
$\Zzk{k-1}\subseteq\Zzk{k}$.

\subsection{Properties of \texorpdfstring{$\Zz$}{Z[ω]}}
\label{ssec:propszz}

We now record some useful properties of $\Zz$. If $u=a+b\omega\in\Zz$,
then
\begin{equation}
  \label{eq:norm}
  u^\dagger u = (a+b\omega)(a+b\omega^2) = a^2 + ab(\omega+\omega^2) +
  b^2 = a^2 - ab + b^2.
\end{equation}
In particular, if $u\in\Zz$, then $u^\dagger u$ is a nonnegative
integer, since the Euclidean norm of a complex number is always
nonnegative.

\begin{definition}
  \label{def:lambda}
  We define $\lambda\in\Zz$ as $\lambda=1-\omega$.
\end{definition}

By \cref{eq:norm}, we have $\lambda^\dagger\lambda = 3$. Similarly, we
have $\lambda^2=1-2\omega+\omega^2= -3\omega$, so that
$3=-\lambda^2\omega^2$. Hence, $3\equiv_\lambda 0$.

\begin{proposition}
  \label{prop:zzetamod}
  We have
  \begin{itemize}
  \item $\Zz/(3) \cong \s{0,1,2, \omega, 2\omega, 1+\omega, 1+2\omega,
    2+\omega, 2+2\omega}\cong \Z/(3) + \omega\Z/(3)$ and
  \item $\Zz/(\lambda) \cong \s{0,1,2}\cong\Z/(3)$.
  \end{itemize}
\end{proposition}

\begin{proof}
  The first item follows from the fact that $3\equiv_3 0$. The second
  item follows from the fact that $3\equiv_\lambda 0$ and the fact
  that $\omega\equiv_\lambda 1$.
\end{proof}

\begin{proposition}
  \label{prop:residues}
  If $u\in\Zz$, then $u^\dagger u \equiv_\lambda 0$ or $u^\dagger u
  \equiv_\lambda 1$.
\end{proposition}

\begin{proof}
  Let $u=a+b\omega\in\Zz$. By \cref{prop:zzetamod},
  $\Zz/(\lambda)\cong \Z/(3)$. By \cref{eq:norm},
  \[
  u^\dagger u = a^2-ab+b^2\equiv_\lambda a^2 +2ab +b^2 = (a+b)^2.
  \]
  Hence $u^\dagger u$ is a square modulo $\lambda$ and therefore
  cannot be congruent to 2, since 0 and 1 are the only squares in
  $\Z/(3)$.
\end{proof}

\begin{proposition}
  \label{prop:residuesquares}
  If $u\in\Zz$, then $u\not\equiv_\lambda 0$ if and only if $u\equiv_3
  \pm \omega^x$ for some $x\in\s{0,1,2}$.
\end{proposition}

\begin{proof}
  The table below lists the elements of $\Zz/(3)$ as given by
  \cref{prop:zzetamod}, together with their residues modulo $\lambda$.
  \begin{center}
  \begin{tabular}[b]{c|c}
  $\Zz/(3)$ & $\Zz/(\lambda)$ \\
  \hline
  $0$ & $0$ \\
  $1$ & $1$ \\
  $2$ & $2$ \\
  $\omega$ & $1$ \\
  $2\omega$ & $2$ \\
  $1+\omega$ & $2$ \\
  $1+2\omega$ & $0$ \\
  $2+\omega$ & $0$ \\
  $2+2\omega$ & $1$ 
  \end{tabular}
  \end{center} 
  The statement then follows by inspection of the table, using the
  fact that $1+\omega = -\omega^2\equiv_3 -\omega^2$ and $2\equiv_3
  -1$.
\end{proof}

\subsection{Denominators}
\label{ssec:denoms}

\begin{definition}
  \label{def:localized}
  Let $k\geq 1$. The ring $\Tzk{k}$ is defined as $\Tzk{k}=\s{u/3^\ell
    \mid u \in \Zzk{k} \mbox{ and } \ell\in\N}$.
\end{definition}

Because the elements of $\Zzk{k}$ can be expressed as linear
combinations of elements of $\Zzk{k-1}$, the elements of $\Tzk{k}$ can
similarly be expressed as linear combinations of elements of
$\Tzk{k-1}$. In particular, for $k\geq 2$, every element $u$ of
$\Tzk{k}$ can be uniquely written as $u=a+b\omega_k+c\omega_k^2$ with
$a,b,c\in\Tzk{k-1}$.

The ring $\Tzk{k}$ is the localization of $\Zzk{k}$ by the powers of
3. Alternatively, $\Tzk{k}$ can be thought of as the localization of
$\Zzk{k}$ by the powers of $\lambda$. Indeed, since
$3=-\omega^2\lambda^2$, we have $3^{-\ell} =
(-\omega^2\lambda^2)^{-\ell}=(-\omega)^\ell(\lambda)^{-2\ell}$. As a
result, any element of $\Tzk{k}$ can be written as $u/\lambda^\ell$
for some $u\in\Zzk{k}$ and some $\ell\in\N$. We leverage this fact to
define, in the usual way (see \cite{Amy2020,Giles2013a,greylyn}), the
notions of \textbf{$\lambda$-denominator exponent} and \textbf{least
  $\lambda$-denominator exponent}.

\begin{definition}
  \label{def:lde}
  Any nonnegative integer $\ell$ such that $v\in\Tzk{k}$ can be
  written as $v=u/\lambda^\ell$ with $u\in \Zzk{k}$ is
  \textbf{$\lambda$-denominator exponent} of $v$. The smallest such
  $\ell$ is the \textbf{least $\lambda$-denominator exponent} of $v$
  and is denoted $\lde(v)$.
\end{definition}

The notions of denominator exponent and least denominator exponent
extend to matrices (and therefore to vectors) with entries in
$\Tzk{k}$: an integer $\ell$ is a $\lambda$-denominator exponent of a
matrix $M$ if it is a $\lambda$-denominator exponent of all of the
entries of $M$; the smallest such $\ell$ is the least
$\lambda$-denominator exponent of $M$.

\subsection{The Group \texorpdfstring{$\unk{k}$}{Un(Z[ζk,1/3]}}
\label{ssec:unzk}

\begin{definition}
  \label{def:unitarygroup}
  We write $\unk{k}$ for the group of $n$-dimensional unitary matrices
  with entries in $\Tzk{k}$ and $\uuk{k}$ for the collection of all
  unitary matrices with entries in $\Tzk{k}$.
\end{definition}

\section{Generators for \texorpdfstring{$\un$}{Un(Z[zeta,1/3]}}
\label{sec:generators}

Following \cite{Amy2020,Giles2013a,greylyn,NC}, we use
\textbf{$m$-level matrices} to define a subset of $\un$ which we will
show to be a generating set.

\begin{definition}
  \label{def:matrices}
  The matrices $(-1)$, $(\omega)$, $X$, and $H$ are defined as
  follows:
  \[
  (-1) = 
  \begin{bmatrix}
  -1
  \end{bmatrix},
  \quad
  (\omega) =
  \begin{bmatrix}
  \omega
  \end{bmatrix},
  \quad
  X =
  \begin{bmatrix}
  0 & 1 \\
  1 & 0
  \end{bmatrix},
  \quad
  \mbox{and}
  \quad
  H = \frac{-\omega^2}{\lambda}
  \begin{bmatrix}
  1 & 1 & 1 \\
  1 & \omega & \omega^2 \\
  1 & \omega^2 & \omega
  \end{bmatrix}.
  \]
\end{definition}

\begin{definition}
  \label{def:onetwolevel}
  Let $M$ be an $m\times m$ matrix, let $m\leq n$, and let $0\leq x_1<
  \ldots< x_m \leq n-1$. The \textbf{$m$-level matrix}
  $M_{[x_1,\ldots, x_m]}$ is the $n\times n$ matrix whose entries are
  given as follows
  \[
  {M_{[x_1,\ldots, x_m]}}_{i,j} = 
  \begin{cases}
  M_{i',j'} & \mbox{if } i=x_{i'} \mbox{ and } j = x_{j'}, \\
  I_{i,j} & \mbox{otherwise.}
  \end{cases}
  \]
\end{definition}

For example, for $n=4$, we have
\[
(\omega)_{[1]} = 
\begin{bmatrix}
1     & \cdot  & \cdot & \cdot \\
\cdot & \omega & \cdot & \cdot \\
\cdot & \cdot  & 1     & \cdot \\
\cdot & \cdot  & \cdot & 1 
\end{bmatrix}
\qquad
\mbox{ and }
\qquad
H_{[0,2,3]} = \frac{-\omega^2}{\lambda}
\begin{bmatrix}
1     & \cdot  & 1        & 1 \\
\cdot & \lambda/(-\omega^2)      & \cdot    & \cdot \\
1     & \cdot  & \omega   & \omega^2 \\
1     & \cdot  & \omega^2 & \omega 
\end{bmatrix}.
\]
When applied to a vector $\ket{u}$, the matrix $(\omega)_{[1]}$ acts
as $(\omega)$ on the entry of index 1 and the matrix $H_{[0,2,3]}$
acts as $H$ on the entries of index $0$, $2$, and $3$.

\begin{definition}
  \label{def:generators}
  We write $\gensn$ for the subset of $\un$ defined as
  \[
  \gensn = \s{(-1)_{[x]}, (\omega)_{[x]}, X_{[x,y]}, H_{[x,y,z]} \mid
    0\leq x<y<z \leq n-1}.
  \]
\end{definition}

\begin{lemma}
  \label{lem:hadamardreduction}
  Let $u_0,u_1,u_2\in\Zz$ be such that $u_0\not\equiv_\lambda 0$,
  $u_1\not\equiv_\lambda 0$, and $u_2\not\equiv_\lambda 0$. Then there
  exists $x_0,x_1,x_2\in\s{0,1,2}$ and $y_0,y_1,y_2\in\s{0,1}$ such
  that
  \[
  H 
  (\omega)_{[0]}^{x_0}(\omega)_{[1]}^{x_1}(\omega)_{[2]}^{x_2} 
  (-1)_{[0]}^{y_0} (-1)_{[1]}^{y_1} (-1)_{[2]}^{y_2}
  \begin{bmatrix}
  u_0 \\ u_1 \\ u_2
  \end{bmatrix}
  =
  \begin{bmatrix}
  u_0' \\ u_1' \\ u_2'
  \end{bmatrix}
  \]
  for some $u_0',u_1',u_2'\in\Zz$ such that $u_0'\equiv_\lambda 0$,
  $u_1'\equiv_\lambda 0$, and $u_2'\equiv_\lambda 0$.
\end{lemma}

\begin{proof}
  Let $j\in\s{0,1,2}$. Since $u_j\not\equiv_\lambda 0$, we have, by
  \cref{prop:residuesquares}, $u_j\equiv_3
  (-1)^{w_j}(\omega)^{z_j}$. Hence, setting $y_j=-w_j$ and $x_j=-z_j$,
  we get $(\omega)^{x_j}(-1)^{y_j}u_j\equiv_3 1$. Therefore,
  \[
  (\omega)_{[0]}^{x_0}(\omega)_{[1]}^{x_1}(\omega)_{[2]}^{x_2} 
  (-1)_{[0]}^{y_0} (-1)_{[1]}^{y_1} (-1)_{[2]}^{y_2}
  \begin{bmatrix}
  u_0 \\ u_1 \\ u_2
  \end{bmatrix}
  =
  \begin{bmatrix}
  v_0 \\ v_1 \\ v_2
  \end{bmatrix}
  \]
  for some $v_0,v_1,v_2\in\Zz$ such that $v_0\equiv_3 v_1\equiv_3
  v_2\equiv_3 1$. The result then follows by computation, since $v_0
  +v_1+v_2 \equiv_3 v_0+ \omega v_1 +\omega^2 v_2 \equiv_3 v_0+
  \omega^2 v_1 + \omega v_2 \equiv_3 0$.
\end{proof}

\begin{lemma}
  \label{lem:base}
  Let $\ket{u}\in \Tz^n$ be a unit vector. If $\lde\ket{u}=0$, then
  $\ket{u}=\pm\omega^x\ket{j}$ for some $0\leq x \leq 2$ and some
  $0\leq j \leq n-1$.
\end{lemma}

\begin{proof}
  Let $\ket{u} \in \Tz^n$. Since $\lde\ket{u}=0$, we have
  $\ket{u}\in\Zz^n$. Since $\ket{u}$ is a unit vector, we also have
  \[
  1 = \braket{u} = \sum u_j^\dagger u_j
  \]
  with $u_j\in\Zz$. Because each $u_j^\dagger u_j$ is a nonnegative
  integer, there must be exactly one $j$ for which $u_j^\dagger
  u_j=1$, while $u_{j'}^\dagger u_{j'}=0$ for all $j'\neq j$. If
  $u_j^\dagger u_j = 1$ then $a_j^2 -a_jb_j +b_j^2=1$, and this
  equation can only be true if $a=\pm 1$ and $b=0$, $a=0$ and $b=\pm
  1$, or $a=b=\pm1$. In the first case, $\ket{u} = \pm \ket{j}$, in
  the second case, $\ket{u} = \pm \omega \ket{j}$, and in the third
  case, $\ket{u} = \pm \omega^2 \ket{j}$,
\end{proof}

\begin{lemma}
  \label{lem:step}
  Let $\ket{u}\in \Tz^n$ be a unit vector. If $\lde\ket{u}>0$, then
  there exists $G_0,\ldots, G_q\in\gensn$ such that $\lde (G_q\cdots
  G_0 \ket{u})<\lde\ket{u}$.
\end{lemma}

\begin{proof}
  Write $\ket{u}$ as $\ket{v}/\lambda^\ell$, with
  $\ell=\lde\ket{u}$. Since $\braket{u}=1$ and $\lambda^\dagger
  \lambda =3$, we get $3^\ell = \braket{v} = \sum v_j^\dagger
  v_j$. Hence, $\sum v_j^\dagger v_j\equiv_\lambda 0$. By
  \cref{prop:residues}, $v_j^\dagger v_j$ is either 0 or 1 modulo
  $\lambda$, and by \cref{prop:zzetamod},
  $\Zz/(\lambda)\cong\Z/(3)$. Thus, the number of $v_j$ such that
  $v_j\not\equiv_\lambda 0$ must be a multiple of 3. Hence, we can
  group the entries of $\ket{v}$ into triples and apply
  \cref{lem:hadamardreduction} to each such triple. This maps
  $\ket{u}$ to some $\ket{u}'$ of lower least denominator exponent.
\end{proof}

\begin{lemma}
  \label{lem:column}
  Let $\ket{u}\in \Tz^n$ be a unit vector and let $0\leq j \leq
  n-1$. Then there exists $G_0,\ldots, G_q\in\gensn$ such that
  $G_q\cdots G_0\ket{u}=\ket{j}$.
\end{lemma}

\begin{proof}
  By induction on $\lde\ket{u}$. If $\lde(\ket{u})=0$, then, by
  \cref{lem:base}, $\ket{u}=\pm\omega^xe_{j'}$ for some $0\leq j' \leq
  n-1$ and some $0\leq x \leq 2$. We can therefore reduce $\ket{u}$ to
  $\ket{j}$ by applying $(-1)_{[j']}$, $(\omega)_{[j']}$, and
  $X_{[j,j']}$ or $X_{[j',j]}$, as needed. If $\lde\ket{u}>0$, then,
  by \cref{lem:step}, there exists $G_p,\ldots, G_0\in\gensn$ such
  that $\lde (G_p\cdots G_0\ket{u}) < \lde(\ket{u})$. We can then
  conclude by applying the induction hypothesis to $G_p\cdots
  G_0\ket{u}$.
\end{proof}

\begin{proposition}
  \label{prop:generation}
  Let $U$ be an $n\times n$ matrix. Then $U\in\un$ if and only if $U$
  can be written as a product of elements of $\gensn$.
\end{proposition}

\begin{proof}
  The right-to-left direction is immediate. For the left-to-right
  direction, consider the matrix $U^\dagger \in \un$. Iteratively
  applying \cref{lem:column} to the columns of $U^\dagger$ yields a
  sequence $G_0, \ldots, G_q$ of elements of $\gensn$ such that
  \[
  G_0 G_1 \cdots Gq U^\dagger = I,
  \]
  and we can therefore write $U$ as $U=G_0G_1\cdots G_q$.
\end{proof}

\section{Exact Synthesis of Toffoli+Hadamard Circuits}
\label{sec:circuits}

Let $\gatesetkone$ be a set of quantum gates. A unitary matrix $U$ can
be \textbf{represented by a circuit over $\gatesetkone$} if there
exists a circuit $C$ over $\gatesetkone$ such that, for any state
$\ket{u}$, we have $C\ket{u} = U\ket{u}$. The circuit may use
ancillary qutrits, but these must start and end the computation in the
same state. If that state can be arbitrary, the ancillary qutrits are
said to be \textbf{borrowed}; if that state is required to be
$\ket{0}$, the ancillary qutrits are said to be \textbf{fresh}. Unless
otherwise specified, ancillae are assumed to be fresh. Note that if a
matrix can be represented by a circuit using $m$ borrowed ancillae,
then it can also be represented by a circuit using $m$ fresh ancillae.

Recall from \cref{sec:intro} that the Clifford-cyclotomic gate set
$\gatesetk{k}$ is defined as $\gatesetk{k} = \s{X, CX, CCX, H,
  \rotk{k}}$. In \cref{app:decompositions} we prove that
$\gatesetk{1}$ is equivalent to the Toffoli+Hadamard gate set, up to
two borrowed ancillae and that, when $k\geq 2$, $\gatesetk{k}$ is
equivalent to the Clifford+$\rotk{k}$ gate set $\s{H, S, CX,
  \rotk{k}}$, up to a single borrowed ancilla. The next proposition
shows that all of the elements of $\gensnn{3^n}$ can be represented by
a circuit over $\gatesetk{1}$ using no more than 2 borrowed
ancillae. The proof of the proposition can be found in
\cref{app:decompositions}.

\begin{proposition}
  \label{prop:circuits}
  If $U\in\gensnn{3^n}$, then $U$ can be represented by a circuit over
  $\gatesetk{1}$ using at most 2 borrowed ancillae.
\end{proposition}

Using \cref{prop:circuits} we are now in a position to define an exact
synthesis algorithm for multiqutrit Toffoli+Hadamard circuits.

\begin{theorem}
  \label{prop:exactsynthbase}
  If $U\in\unkk{3^n}{1}$, then $U$ can be represented by an $n$-qutrit
  circuit over $\gatesetk{1}$ using at most 2 ancillae.
\end{theorem}

\begin{proof}
  By \cref{prop:generation}, $\gensnn{3^n}$ generates
  $\unkk{3^n}{1}$. Hence, it is sufficient to show that the elements
  of $\gensnn{3^n}$ can be represented by an $n$-qutrit circuit over
  $\gatesetk{1}$. This follows from \cref{prop:circuits}, since if 2
  borrowed ancillae suffice to construct a circuit for $U$, then 2
  fresh ancillae are also sufficient for this purpose.
\end{proof}

\section{Catalytic Embeddings}
\label{sec:catembs}

\begin{definition}
  \label{def:catemb}
  Let $\calu$ and $\calv$ be collections of unitaries. An
  \textbf{$m$-dimensional catalytic embedding} of $\calu$ into $\calv$
  is a pair $(\phi,\ket{c})$ of a function $\phi:\calu \to \calv$ and
  a vector $\ket{c}\in\C^m$ such that if $U\in\calu$ has dimension $n$
  then $\phi(U)\in\calv$ has dimension $nm$ and
  \[
  \phi(U) (\ket{u}\otimes\ket{c}) = (U\ket{u})\otimes \ket{c}
  \]
  for every $\ket{u}\in\C^n$. The vector $\vecc$ is the
  \textbf{catalyst} of the catalytic embedding $(\phi,\ket{c})$. We
  sometimes express the fact that $(\phi,\ket{c})$ is a catalytic
  embedding of $\calu$ into $\calv$ by writing
  $(\phi,\ket{c}):\calu\to\calv$.
\end{definition}

\begin{definition}
  \label{def:concat}
  Let $(\phi,\ket{c}):\calu\to\calv$ and
  $(\phi',\ket{c}'):\calv\to\calw$ be catalytic embeddings of
  dimension $m$ and $m'$, respectively. The \textbf{concatenation} of
  $(\phi,\ket{c})$ and $(\phi',\ket{c}')$ is the $m'm$-dimensional
  catalytic embedding $(\phi',\ket{c}')\circ(\phi,\ket{c})$ defined by
  $(\phi',\ket{c}')\circ(\phi,\ket{c}) = (\phi'\circ
  \phi,\ket{c}\otimes\ket{c}')$.
\end{definition}

The concatenation of catalytic embeddings is associative and the
catalytic embedding $(1_{\calu},[1]):\calu\to\calu$ acts as an
identity for concatenation.

\begin{definition}
  \label{def:pcm}
  Let $k\geq 2$. We define $\Omega_k$ and $\ket{c_k}$ as
  \[
  \Omega_k = 
  \begin{bmatrix}
  \cdot & \cdot & \omega_{k-1} \\
  1     & \cdot & \cdot \\
  \cdot & 1     & \cdot
  \end{bmatrix}
  \qquad
  \mbox{ and }
  \qquad
  \ket{c_k}=\frac{1}{\lambda}
  \begin{bmatrix}
  1 \\
  \omega_{k}^{-1} \\
  \omega_{k}^{-2} \\
  \end{bmatrix}.
  \]
\end{definition}

The matrix $\Omega_k$ is unitary and the vector $\ket{c_k}$ is an
eigenvector of $\Omega_k$ for eigenvalue $\omega_{k}$. Indeed, since
$\omega_{k-1}=\omega_{k}^3$, we have
\begin{equation}
\label{eq:eigenvector}
\Omega_k \ket{c_k} = \frac{1}{\lambda}
\begin{bmatrix}
\cdot & \cdot & \omega_{k-1} \\
1     & \cdot & \cdot \\
\cdot & 1     & \cdot
\end{bmatrix}
\begin{bmatrix}
1 \\
\omega_{k}^{-1} \\
\omega_{k}^{-2} \\
\end{bmatrix} =
\frac{1}{\lambda}
\begin{bmatrix}
\omega_k \\
1 \\
\omega_k^{-1}
\end{bmatrix} =
\frac{\omega_k}{\lambda}
\begin{bmatrix}
1 \\
\omega_k^{-1} \\
\omega_k^{-2}
\end{bmatrix} = \omega_k\ket{c_k}.
\end{equation}

Now recall from \cref{ssec:denoms} that, for $k\geq 2$, every
$u\in\Tzk{k}$ can be written uniquely as a linear combination of the
form $u=a+b\omega_k +c\omega_k^2$, where
$a,b,c\in\Tzk{k-1}$. Therefore, every matrix $U$ over $\Tzk{k}$ can be
uniquely written as $U = A + B\omega_{k} + C\omega_{k}^2$, where $A$,
$B$, and $C$ are matrices over $\Tzk{k-1}$. We use this fact below to
define a function $\uuk{k} \to \uuk{k-1}$.

\begin{proposition}
  \label{prop:catembfunc}
  Let $k\geq 2$. The assignment
  \[
  A + B\omega_{k} + C\omega_{k}^2 \longmapsto A\otimes I + B\otimes
  \Omega_{k} + C\otimes \Omega_{k}^2
  \]
  defines a function $\phi_k:\uuk{k} \to \uuk{k-1}$.
\end{proposition}

\begin{proof}
  Let $U\in\uuk{k}$ and write $U$ as $U=A + B\omega_{k}+C\omega_{k}^2$
  for some matrices $A$, $B$, and $C$ over $\Tzk{k-1}$. Now let
  $U'=A\otimes I + B\otimes \Omega_{k}+ C\otimes \Omega_{k}^2$. It is
  clear that $U'$ is a matrix with entries in $\Tzk{k-1}$. We now show
  that $U'$ is unitary. Since $U$ is unitary and since $U=A +
  B\omega_{k}+C\omega_{k}^2$, we can express the equation $U^\dagger
  U=I$ in terms of $A$, $B$, and $C$. Using $\omega_k^\dagger
  =\omega_{k-1}^\dagger\omega_k^2$, this yields
  \[
  (A^\dagger A + B^\dagger B + C^\dagger C)+ (A^\dagger B + B^\dagger
  C + C^\dagger A\omega_{k-1}^\dagger)\omega_{k} + (A^\dagger C +
  B^\dagger A\omega_{k-1}^\dagger + C^\dagger
  B\omega_{k-1}^\dagger)\omega_{k}^2=I.
  \]
  Hence, $A^\dagger A + B^\dagger B + C^\dagger C=I$ and $A^\dagger B
  + B^\dagger C + C^\dagger A\omega_{k-1}^\dagger=A^\dagger C +
  B^\dagger A\omega_{k-1}^\dagger + C^\dagger
  B\omega_{k-1}^\dagger=0$. Now note that $\Omega_k^\dagger =
  \omega_{k-1}^\dagger\Omega_k^2$, so that $U'^\dagger U'$ is equal to
  \[
  (A^\dagger A + B^\dagger B + C^\dagger C)\otimes I+ (A^\dagger B +
  B^\dagger C + C^\dagger A\omega_{k-1}^\dagger)\otimes \Omega_{k} +
  (A^\dagger C + B^\dagger A\omega_{k-1}^\dagger + C^\dagger
  B\omega_{k-1}^\dagger)\otimes\Omega_{k}^2.
  \]
  Hence, $U'^\dagger U'=I$. Reasoning analogously shows that
  $U'U'^\dagger = I$, so that $U'$ is indeed unitary.
\end{proof}

\begin{proposition}
  \label{prop:catembkmone}
  Let $k\geq 2$. The pair $(\phi_k,\ket{c_k})$ is a 3-dimensional
  catalytic embedding of $\uuk{k}$ into $\uuk{k-1}$.
\end{proposition}

\begin{proof}
  By \cref{prop:catembfunc}, $\phi_k:\uuk{k} \to \uuk{k-1}$ is a
  function and, by construction, if $U\in\uuk{k}$ has dimension $n$,
  then $\phi_k(U)$ has dimension $3n$. Moreover, if $\ket{u}\in\C^n$,
  then
  \begin{align*}
  \phi_k(U) (\ket{u}\otimes\ket{c_k}) &= (A\otimes I + B \otimes \Omega_{k} + C\otimes \Omega_{k}^2) (\ket{u}\otimes\ket{c_k})\\
   &= A\otimes I (\ket{u}\otimes\ket{c_k})+ B \otimes \Omega_{k}(\ket{u}\otimes\ket{c_k})+ C \otimes \Omega_{k}^2(\ket{u}\otimes\ket{c_k})\\
   &= A\ket{u}\otimes I \ket{c_k}+ B\ket{u} \otimes \Omega_{k}\ket{c_k} + C\ket{u} \otimes \Omega_{k}^2\ket{c_k} \\
  &= A\ket{u}\otimes \ket{c_k}+ B\ket{u} \otimes \omega_{k}\ket{c_k} + C\ket{u} \otimes \omega_{k}^2\ket{c_k}\\ 
  &= A\ket{u}\otimes \ket{c_k}+ \omega_{k}B\ket{u} \otimes \ket{c_k} + \omega_{k}^2C\ket{u} \otimes \ket{c_k}\\ 
  &= (A\ket{u}+ \omega_{k}B\ket{u}+ \omega_{k}^2C\ket{u})\otimes\ket{c_k} \\ 
  &=(U\ket{u})\otimes \ket{c_k}. 
  \end{align*}
  Hence, $(\phi_k,\ket{c_k})$ is a catalytic embedding.
\end{proof}

\begin{remark}
  The catalytic embedding constructed in \cref{prop:catembfunc} and
  \cref{prop:catembkmone} takes advantage of the fact that the matrix
  $\Omega_k$ and the algebraic number $\omega_k$ have many properties
  in common. Importantly, the polynomial $x^3 - \omega_{k-1}$ is both
  the characteristic polynomial of $\Omega_k$ and the minimal
  polynomial of $\omega_k$ over the ring $\Tzk{k-1}$. This
  construction generalizes to many other rings of interest (see
  \cite{catemb}).
\end{remark}

\begin{corollary}
  \label{cor:concatemb}
  Let $k\geq 2$. There is a $3^{k-1}$-dimensional catalytic embedding
  $(\phi,\vecc):\uuk{k}\to\uu$.
\end{corollary}

\begin{proof}
  Applying \cref{prop:catembkmone} repeatedly yields a sequence of
  catalytic embeddings
  \[
  \uuk{k} \xrightarrow{(\phi_k,\ket{c_k})}\cdots
  \xrightarrow{(\phi_3,\ket{c_3})}\uuk{2}\xrightarrow{(\phi_2,\ket{c_2})}
  \uu.
  \]
  Concatenating the catalytic embeddings in this sequence yields the
  desired result.
\end{proof}

Note that the catalyst $\vecc$ in the catalytic embedding
$(\phi,\vecc)$ of \cref{cor:concatemb} is the product state
$\vecc=\ket{c_2}\otimes \cdots \otimes \ket{c_k}$.

\section{Exact Synthesis of Clifford-Cyclotomic Circuits}
\label{sec:exactsynth}

We can now prove our main result, which will follow straightforwardly
from the results of \cref{sec:generators,sec:catembs,sec:circuits}.

\begin{proposition}
  \label{prop:exactsynthstep}
  Let $k\geq 2$. If $U\in\unkk{3^n}{k}$, then $U$ can be represented
  by an $n$-qutrit circuit over $\gatesetk{k}$ using at most $k+1$
  ancillae.
\end{proposition}

\begin{proof}
  Let $U\in\unkk{3^n}{k}$ and let $(\phi,\vecc)$ be the catalytic
  embedding constructed in \cref{cor:concatemb}, with
  $\vecc=\ket{c_2}\otimes \cdots \otimes \ket{c_k}$. We then have
  $\phi(U)\in\unn{3^{n+k-1}}$, so that, by \cref{prop:exactsynthbase},
  $\phi(U)$ can be represented by an $(n+k-1)$-qutrit circuit $C$ over
  $\gatesetk{1}$ using at most 2 fresh ancillae. By \cref{def:catemb},
  the action of $\phi(U)$ on an input of the form
  $\ket{u}\otimes\ket{c_2}\otimes \cdots \otimes \ket{c_k}$ can be
  depicted as below (where the ancillary qutrits used in $C$, if any,
  are omitted).
  
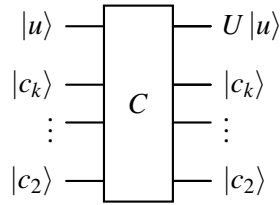
\begin{center}
  \begin{quantikz}
  \lstick{$\ket{u}$}   & \gate[wires=4,nwires=3]{~~C~~} & \qw\rstick{$U\ket{u}$}\\
  \lstick{$\ket{c_k}$} &                                    & \qw\rstick{$\ket{c_k}$}\\
  \lstick{$\vdots$}  &                                    & \rstick{$\vdots$}\\ 
  \lstick{$\ket{c_2}$} &                                    & \qw\rstick{$\ket{c_2}$}      
  \end{quantikz}    
  \end{center}
  But, for $2\leq \ell \leq k$, we have $\ket{c_\ell} = T_\ell^\dagger
  H \ket{0}$ and $T_\ell^\dagger = (T_k^\dagger)^{3^{k-\ell}}$. Hence,
  we can construct the following circuit over $\gatesetk{k}$.
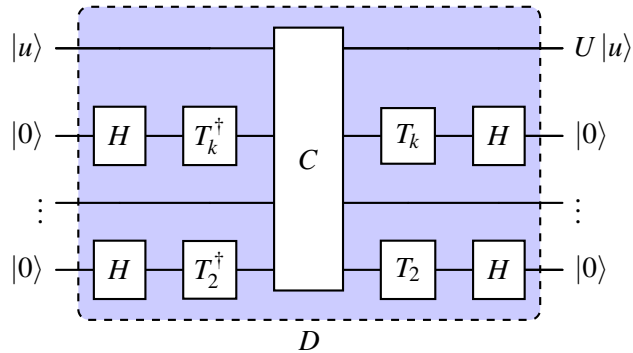
  \begin{center}
  \begin{quantikz}
  \lstick{$\ket{u}$}  & \qw \gategroup[4,steps=5,style={dashed,rounded corners,fill=blue!20, inner xsep=2pt},background,label style={label position=below,anchor=north,yshift=-0.2cm}]{{$D$}} &\qw & \gate[wires=4,nwires=3]{~~C~~} & \qw & \qw & \qw\rstick{$U\ket{u}$}\\
  \lstick{$\ket{0}$} & \gate{H}\vphantom{T_k^\dagger}\hphantom{~} & \gate{T_k^\dagger} & & \gate{T_k}\vphantom{T_k^dagger}\hphantom{~} & \gate{H}\vphantom{T_k^\dagger}\hphantom{~} & \qw\rstick{$\ket{0}$}\\
  \lstick{$\vdots$}  & & &                                    & & & \rstick{$\vdots$}\\ 
  \lstick{$\ket{0}$} & \gate{H}\vphantom{T_k^\dagger}\hphantom{~} & \gate{T_2^\dagger} &                    & \gate{T_2}\vphantom{T_k^\dagger}\hphantom{~} & \gate{H}\vphantom{T_k^\dagger}\hphantom{~}               & \qw\rstick{$\ket{0}$}      
  \end{quantikz}    
  \end{center}
  Since all of the ancillae in $D$ (including the ancillae potentially
  present in $C$) start and end the computation in the $\ket{0}$
  state, then $D$ is a circuit over $\gatesetk{k}$ which represents
  $U$ and uses at most $k+1$ (fresh) ancillae, as desired.
\end{proof}

\begin{remark}
  The circuit constructed in \cref{prop:exactsynthstep} actually use
  $k-1$ fresh ancillae and no more than 2 borrowed ancillae. For
  brevity, we simply stated the proposition in terms of fresh
  ancillae. One can amend the constructions in 
  \cref{app:decompositions} to reduce the total ancilla-count from 
  $k+1$ to $k$, at the cost of requiring all ancillae to be fresh.
\end{remark}

\begin{theorem}
  \label{thm:characterization}
  Let $k\geq 1$ and let $U$ be a $3^n\times 3^n$ unitary matrix. Then
  $U$ can be represented by an $n$-qutrit circuit over $\gatesetk{k}$
  if and only if $U\in\unkk{3^n}{k}$. Moreover, $k+1$ ancillae are
  always sufficient to construct a circuit for $U$.
\end{theorem}

\begin{proof}
  The left-to-right direction is a consequence of the fact that the
  entries of the elements of $\gatesetk{k}$ lie in the ring
  $\Tzk{k}$. The right-to-left direction is given by
  \cref{prop:exactsynthbase,prop:exactsynthstep}.
\end{proof}

\section{Circuit Complexity}
\label{sec:complexity}

The proof of \cref{thm:characterization} is constructive: it provides
an algorithm to construct a circuit for a given matrix. In this
section, we briefly discuss the complexity of the resulting circuit,
reasoning as in \cite{li2023,Giles2013a}. We start by considering
\cref{prop:generation} before turning to \cref{thm:characterization}.

\begin{lemma}
  \label{lem:complexitybase}
  Let $U\in \unn{m}$ and let $\ell = \lde(U)$. The algorithm of
  \cref{prop:generation} expresses $U$ as a product of $O(2^m\ell)$
  elements of $\gensnn{m}$ in the worst case.
\end{lemma}

\begin{proof}
  Consider the first column of $U$. In the worst case, its least
  denominator exponent is $\ell$. To reduce this least denominator
  exponent by one requires $O(m)$ operations. Hence, reducing the
  first column of $U$ completely requires $O(\ell m)$ operations in
  the worst case. The reduction of the first column may increase the
  least denominator exponent of the second column from $\ell$ to
  $2\ell$, since each entry of the second column may be affected by up
  to $\ell$ 3-level matrices in the course of this first
  reduction. Once the first column has been reduced, the second column
  may still have $m-1$ nonzero entries. Reducing the second column
  will hence require $O(2\ell(m-1))$ operations in the worst case. In
  general, reducing the $j$-th column will require
  $O(2^{j-1}\ell(m-j))$ operations in the worst case so that the
  overall reduction of $U$ requires at most
  \[
  O\left(\sum_{i=0}^{n-1} 2^i \ell (m-i)\right)
  \]
  operations. Simplifying the resulting sum yields a total of $O(2^m
  \ell)$ operations.
\end{proof}

\begin{theorem}
  \label{thm:complexitystep}
  Let $U\in\unkk{3^n}{k}$ and let $\ell = \lde(U)$. The algorithm of
  \cref{thm:characterization} represents $U$ as a circuit of
  $O((n+k)2^{3^{n+k-1}}\ell)$ gates in the worst case.
\end{theorem}

\begin{proof}
  The algorithm of \cref{thm:characterization} uses the catalytic
  embedding $(\phi,\vecc)$ of \cref{cor:concatemb} to construct a
  matrix $\phi(U)$ over $\Tz$. The dimension of $\phi(U)$ is
  $3^{n+k-1}$ and its least denominator exponent is no more than
  $\ell$. Hence, by \cref{lem:complexitybase}, the algorithm of
  \cref{prop:generation} will express $\phi(U)$ as a product of no
  more than $O(2^{3^{n+k-1}}\ell)$ elements of
  $\gensnn{3^{n+k-1}}$. It follows from the circuit constructions
  given in \cref{app:decompositions}, that each element of
  $\gensnn{3^{n+k-1}}$ can be represented by a circuit consisting of
  $O(n+k)$ gates. Hence, the circuit produced by
  \cref{thm:characterization} consists of no more than
  $O((n+k)2^{3^{n+k-1}}\ell)$ gates.
\end{proof}

\section{Conclusion}
\label{sec:conc}

We showed that the matrices that can be exactly represented by an
$n$-qutrit circuit over the Clifford-cyclotomic gate set of degree
$3^k$ are precisely the elements of $\unkk{3^n}{k}$. Moreover, we
showed that no more than $k+1$ ancillae are required to construct a
circuit for an element of $\unkk{3^n}{k}$.

Our proof contains an algorithm for synthesizing a circuit over
$\gatesetk{k}$, given a matrix in $\unkk{3^n}{k}$. However, the
circuits constructed in this way are very large and their optimization
is a promising direction for future research. It would be interesting
to reduce the gate-complexity of the circuits produced by
\cref{thm:characterization}. The techniques employed in
\cite{li2023,kliuchnikov2013synthesis} for the synthesis of multiqubit
Toffoli+Hadamard and Clifford+$T$ circuits are likely to apply in the
qutrit context as well. Similarly, it would also be interesting to
reduce the number of ancillae used by the algorithm. As
\cref{app:decompositions} shows, some of the ancillae can be removed
by choosing a slightly different gate set, but the bulk of the
ancillae come from the use of catalytic embeddings, so a different
synthesis technique may be required for more significant
savings. Along this line of inquiry, it would be interesting to
characterize the matrices that can be represented by ancilla-free
circuits. Such characterizations exist for qubit matrices
\cite{Amy2020,Giles2013a}, but are likely to be different for qutrits
\cite{liaRC}.

Finally, a natural generalization of this work would be to consider
higher-dimensional qudits. However, preliminary research suggests that
the techniques used here and in \cite{qubitcyclo} might not adapt
straightforwardly to primes larger than 3.  While it stands to reason
that some version of our results should continue to hold for larger
prime dimensions, proving this to be the case might require new ideas.

\subsection*{Acknowledgements}

The authors would like to thank Sarah Meng Li, Ewan Murphy, and the
anonymous reviewers of 21st International Conference on Quantum
Physics and Logic (QPL 2024) for insightful comments on an earlier
version of this paper. LY is funded by a Google PhD Fellowship. The
circuit diagrams in the proof of \cref{prop:exactsynthstep} were
typeset using Quantikz \cite{quantikz}.

\bibliographystyle{eptcs}
\bibliography{multiqutrit}

\newpage

\appendix

\section{Circuit Constructions}
\label{app:decompositions}

In this appendix, we show that the Clifford-cyclotomic gate set
$\gatesetk{k}$ is equivalent to the Clifford+$T_k$ gate set when
$k\geq 2$, we give a construction of the $CX$ gate over the $\s{X,
  CCX, H}$ gate set, and we provide a proof of
\cref{prop:circuits}. In addition, we show that the matrices
$(-1)_{[x]}$, $(\omega_k)_{[x]}$, $X_{[x_1,x_2]}$, and
$H_{[x_1,x_2,x_3]}$ can be represented by circuits over the
$\gatesetk{k}$ gate set using at most $k$ borrowed ancillae. The
constructions in this appendix are exact (i.e., not up to a global or
relative phase). Implementations of our constructions, for a fixed
number of controls, are available at
\url{https://github.com/lia-approves/qutrit-Clifford-cyclotomic}.

\subsection{Gate Set Equivalences}
\label{ssec:equivs}

Recall from \cref{sec:intro} that the qutrit Toffoli (or $CCX$) gate 
acts on computational basis states as
\[
  \ket{x,y,z} \mapsto \ket{x,y,z+xy},
\]
where the arithmetic operations are performed modulo 3. In higher 
prime dimension $d$, the Toffoli gate is defined similarly, except 
that the arithmetic operations are performed modulo $d$. The 
Toffoli gate can be represented in the qupit ZH-calculus 
\cite{qupitZH} as below.
\begin{equation}
  \label{eq:qupittof}
  \begin{tikzpicture}
	\begin{pgfonlayer}{nodelayer}
		\node [style=none] (0) at (-2.5, 0) {};
		\node [style=none] (1) at (-2.5, -1) {};
		\node [style=smallbox] (3) at (-1.75, -1) {\tiny X};
		\node [style=none] (4) at (-1, 0) {};
		\node [style=none] (5) at (-1, -1) {};
		\node [style=none] (6) at (-2.5, 1) {};
		\node [style=none] (8) at (-1, 1) {};
		\node [style=none] (9) at (0, 0) {$\leftrightarrow$};
		\node [style=none] (10) at (1, 0) {};
		\node [style=none] (11) at (1, -1) {};
		\node [style=Z dot] (12) at (1.5, 0) {};
		\node [style=rn] (13) at (4.125, -1) {};
		\node [style=none] (14) at (5.125, 0) {};
		\node [style=none] (15) at (5.125, -1) {};
		\node [style=none] (16) at (1, 1) {};
		\node [style=Z dot] (17) at (1.5, 1) {};
		\node [style=none] (18) at (5.125, 1) {};
		\node [style=hadamard] (19) at (2.25, 0.5) {};
		\node [style=none] (20) at (3.625, 0.5) {};
		\node [style=small hadamard] (21) at (2.675, 0.5) {};
		\node [style=small hadamard] (22) at (3.05, 0.5) {};
		\node [style=small hadamard] (23) at (3.425, 0.5) {};
		\node [style=rn] (24) at (4.625, -1) {};
		\node [style=0 control] (25) at (-1.75, 1) {$\Lambda$};
		\node [style=0 control] (26) at (-1.75, 0) {$\Lambda$};
	\end{pgfonlayer}
	\begin{pgfonlayer}{edgelayer}
		\draw [style=none] (5.center) to (3);
		\draw [style=none] (3) to (1.center);
		\draw [style=none] (10.center) to (12);
		\draw [style=none] (12) to (14.center);
		\draw [style=none] (15.center) to (13);
		\draw [style=none] (13) to (11.center);
		\draw [style=none] (16.center) to (17);
		\draw [style=none] (17) to (18.center);
		\draw [style=none, in=150, out=-30, looseness=1.25] (17) to (19);
		\draw [style=none, in=-150, out=30, looseness=1.25] (12) to (19);
		\draw [style=none] (19) to (20.center);
		\draw [style=none, in=135, out=0, looseness=0.75] (20.center) to (13);
		\draw (6.center) to (25);
		\draw (25) to (8.center);
		\draw (4.center) to (26);
		\draw (0.center) to (26);
		\draw (3) to (26);
		\draw (26) to (25);
	\end{pgfonlayer}
\end{tikzpicture}
\end{equation}
In \cref{eq:qupittof}, $\Lambda$ denotes the following type of 
control: if $U$ is a unitary and $\ket{c}$ and $\ket{t}$ are 
computational basis states, then 
$\Lambda(U)\ket{c}\ket{t} = \ket{c}\otimes (U^c\ket{t})$. In 
particular, $\Lambda(X)$ is the $CX$ gate and 
$\Lambda(\Lambda(X))=\Lambda(CX)$ is the $CCX$ gate.

We now recall the definition of the \textbf{$\ket{0}$-controlled $X$}
gate, which applies an $X$ gate to its target if and only if its control
is in the state $\ket{0}$ \cite{qupitZH}.

\begin{definition}
  Let $d$ be a prime. The qudit \textbf{$\ket{0}$-controlled $X$} gate
  acts on computational basis states as
  \[
      \ket{c, t} \ \mapsto\ 
       \begin{cases}
          \ket{c,t+1} \ &\text{if}\ c = 0, \mbox{ and}\\
          \ket{c,t} \ &\text{otherwise,}
      \end{cases}
  \]
  where arithmetic is performed modulo $d$.
\end{definition}

Remarkably, when $d$ is a prime greater than 2, the $X$ gate and the
$\ket{0}$-controlled $X$ gate suffices to generate all of the $d$-ary
classical reversible gates \cite{qupitZH}. Moreover, as was shown in
\cite{selinger_odd,liaRC}, when $d=3$, no ancillary qutrits are needed
for this purpose. In contrast, there is no collection of reversible
one and two-qubit gates that suffices to generate all of the binary
reversible gates.

\begin{theorem}[\cite{liaRC}, Theorem~2]
  \label{thm:daryrev}
  Any ternary classical reversible function $f:\s{0,1,2}^n\to
  \s{0,1,2}^n$ can be represented by an ancilla-free circuit of $X$
  and $\ket{0}$-controlled $X$ gates.
\end{theorem}

Here, we only need multiply-controlled Toffoli gates, which can be
built with a gate count linear in the number of controls, as in
\cite{qupitZH,ZiW2023optsynthmultictrlqudit}. The constructions of
\cite{qupitZH,ZiW2023optsynthmultictrlqudit} use no more borrowed
ancillae than there are controls. They can be made into ancilla-free
constructions by building Toffoli gates with $n/2$ controls using at
most $n/2$ borrowed ancillae. Following \cite{liaRC}, one can then
combine six of these Toffoli gates with $n/2$ controls to construct a
Toffoli gate with $n-1$ controls, and then combine 3 of these Toffoli
gates with $n-1$ controls to add the final control.

We now show that the $CX$ gate can be represented by a circuit over
$\s{X, CCX, H}$.

\begin{lemma}
  \label{lem:constructCX}
  The gate sets $\s{X, CX, CCX, H}$ and $\s{X, CCX, H}$ are equivalent
  up to a single borrowed ancilla.
\end{lemma}

\begin{proof}
  The circuit below represents the $CX$ gate using a single borrowed
  ancilla.
  
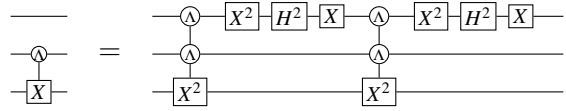
\[
      \begin{tikzpicture}
	\begin{pgfonlayer}{nodelayer}
		\node [style=none] (0) at (-4.25, -1) {};
		\node [style=0 control] (1) at (-3.25, 0) {\tiny$\Lambda$};
		\node [style=0 control] (2) at (-3.25, 1) {\tiny$\Lambda$};
		\node [style=none] (3) at (-4.25, 1) {};
		\node [style=none] (4) at (-4.25, 0) {};
		\node [style=smallbox] (5) at (-3.25, -1) {$X^2$};
		\node [style=neg0 control] (6) at (-7.25, 0) {\tiny$\Lambda$};
		\node [style=smallbox] (7) at (-7.25, -1) {$X$};
		\node [style=none] (8) at (-8, 0) {};
		\node [style=none] (9) at (-6.5, 0) {};
		\node [style=none] (10) at (-6.5, -1) {};
		\node [style=none] (11) at (-8, -1) {};
		\node [style=none] (12) at (-5.375, 0) {$=$};
		\node [style=0 control] (13) at (1.75, 0) {\tiny$\Lambda$};
		\node [style=0 control] (14) at (1.75, 1) {\tiny$\Lambda$};
		\node [style=smallbox] (15) at (1.75, -1) {$X^2$};
		\node [style=smallbox] (16) at (-1.875, 1) {$X^2$};
		\node [style=none] (17) at (6.625, -1) {};
		\node [style=none] (18) at (6.625, 1) {};
		\node [style=none] (19) at (6.625, 0) {};
		\node [style=none] (20) at (-8, 1) {};
		\node [style=none] (21) at (-6.5, 1) {};
		\node [style=smallbox] (22) at (-0.625, 1) {$H^2$};
		\node [style=smallbox] (23) at (0.5, 1) {$X$};
		\node [style=smallbox] (24) at (3.125, 1) {$X^2$};
		\node [style=smallbox] (25) at (4.375, 1) {$H^2$};
		\node [style=smallbox] (26) at (5.5, 1) {$X$};
	\end{pgfonlayer}
	\begin{pgfonlayer}{edgelayer}
		\draw (0.center) to (5);
		\draw (1) to (5);
		\draw (1) to (2);
		\draw [style=none] (8.center) to (6);
		\draw [style=none] (6) to (9.center);
		\draw [style=none] (10.center) to (7);
		\draw [style=none] (7) to (11.center);
		\draw (13) to (15);
		\draw (13) to (14);
		\draw (5) to (15);
		\draw (15) to (17.center);
		\draw (21.center) to (20.center);
		\draw (6) to (7);
		\draw (3.center) to (2);
		\draw (2) to (16);
		\draw (22) to (16);
		\draw (22) to (23);
		\draw (23) to (14);
		\draw (14) to (24);
		\draw (24) to (25);
		\draw (25) to (26);
		\draw (26) to (18.center);
		\draw (19.center) to (13);
		\draw (13) to (1);
		\draw (1) to (4.center);
	\end{pgfonlayer}
\end{tikzpicture}
  \]
\end{proof}

\begin{proposition}
  \label{prop:construct-czerox}
  Let $C_{\ket{0}}X$ denote the qutrit $\ket{0}$-controlled $X$
  gate. Then the gate sets $\s{X, CCX, H}$ and $\s{X,C_{\ket{0}}X,H}$
  are equivalent up to a single borrowed ancilla.
\end{proposition}

\begin{proof}
  The gates $X$ and $CCX$ are ternary classical reversible
  functions. Hence, by \cref{thm:daryrev}, they can both be
  represented by a circuit over $\s{X,C_{\ket{0}}X,H}$. Thus, every
  matrix that can be represented by a circuit over $\s{X, CCX, H}$ can
  be represented by a circuit over $\s{X,C_{\ket{0}}X,H}$. Conversely,
  we have
  
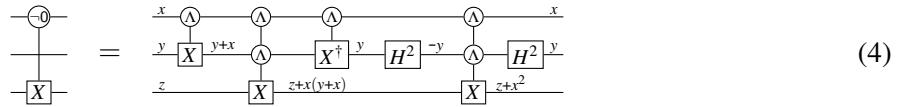
\begin{equation}\label{eq:zcx-circ-decomp-bor-pf}
      \begin{tikzpicture}
	\begin{pgfonlayer}{nodelayer}
		\node [style=none] (11) at (1.125, -1) {};
		\node [style=0 control] (12) at (4, 0) {\tiny$\Lambda$};
		\node [style=0 control] (17) at (4, 1) {\tiny$\Lambda$};
		\node [style=none] (25) at (1.125, 1) {};
		\node [style=0 control] (26) at (2.125, 1) {\tiny$\Lambda$};
		\node [style=none] (63) at (1.125, 0) {};
		\node [style=smallbox] (69) at (2.125, 0) {$X$};
		\node [style=smallbox] (70) at (4, -1) {$X$};
		\node [style=0 control] (77) at (5.875, 1) {\tiny$\Lambda$};
		\node [style=smallbox] (80) at (5.875, 0) {$X^{\dagger}$};
		\node [style=neg0 control] (81) at (-1.875, 1) {\tiny$\neg$\textsf{0}};
		\node [style=smallbox] (82) at (-1.875, -1) {$X$};
		\node [style=none] (83) at (-2.625, 1) {};
		\node [style=none] (84) at (-1.125, 1) {};
		\node [style=none] (85) at (-1.125, -1) {};
		\node [style=none] (86) at (-2.625, -1) {};
		\node [style=none] (87) at (0, 0) {$=$};
		\node [style=smallbox] (88) at (7.75, 0) {$H^2$};
		\node [style=0 control] (89) at (9.625, 0) {\tiny$\Lambda$};
		\node [style=0 control] (90) at (9.625, 1) {\tiny$\Lambda$};
		\node [style=smallbox] (93) at (9.625, -1) {$X$};
		\node [style=smallbox] (94) at (11, 0) {$H^2$};
		\node [style=none] (95) at (12, -1) {};
		\node [style=none] (96) at (12, 1) {};
		\node [style=none] (98) at (12, 0) {};
		\node [style=none] (99) at (1.375, 1.175) {\tiny$x$};
		\node [style=none] (101) at (1.375, 0.175) {\tiny$y$};
		\node [style=none] (103) at (1.375, -0.825) {\tiny$z$};
		\node [style=none] (105) at (3, 0.2) {\tiny$y$+$x$};
		\node [style=none] (106) at (5.5, -0.8) {\tiny$z$+$x(y$+$x)$};
		\node [style=none] (108) at (6.625, 0.2) {\tiny$y$};
		\node [style=none] (109) at (8.625, 0.2) {\tiny $\texttt{-}y$};
		\node [style=none] (110) at (10.625, -0.75) {\tiny$z$+$x^2$};
		\node [style=none] (113) at (11.75, 1.175) {\tiny$x$};
		\node [style=none] (114) at (11.75, 0.175) {\tiny$y$};
		\node [style=none] (116) at (-2.625, 0) {};
		\node [style=none] (117) at (-1.125, 0) {};
	\end{pgfonlayer}
	\begin{pgfonlayer}{edgelayer}
		\draw (25.center) to (26);
		\draw (26) to (17);
		\draw (63.center) to (69);
		\draw (69) to (26);
		\draw (69) to (12);
		\draw (80) to (77);
		\draw (11.center) to (70);
		\draw (12) to (70);
		\draw (12) to (17);
		\draw (17) to (77);
		\draw (80) to (12);
		\draw [style=none] (82) to (81);
		\draw [style=none] (83.center) to (81);
		\draw [style=none] (81) to (84.center);
		\draw [style=none] (85.center) to (82);
		\draw [style=none] (82) to (86.center);
		\draw (89) to (93);
		\draw (89) to (90);
		\draw (70) to (93);
		\draw (93) to (95.center);
		\draw (98.center) to (94);
		\draw (94) to (89);
		\draw (89) to (88);
		\draw (88) to (80);
		\draw (77) to (90);
		\draw (90) to (96.center);
		\draw (117.center) to (116.center);
	\end{pgfonlayer}
\end{tikzpicture}
  \end{equation}
  where $x,y,z\in \{0,1,2\}$ are input qutrit computational basis
  states and the basis state on a wire is updated whenever it is
  changed by the circuit. The $\neg 0$ on the left-hand side of
  \cref{eq:zcx-circ-decomp-bor-pf} indicates that the $X$ gate is
  applied when the control is not in the state $\ket{0}$. To see that
  \cref{eq:zcx-circ-decomp-bor-pf} holds, note that $x^2=1$ for $x\neq
  0$ so that $z+x^2$ is indeed the desired state. Moreover, we have
  
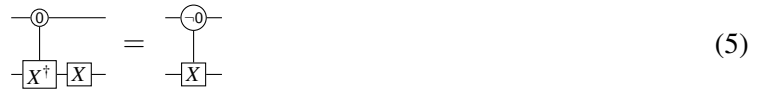
\begin{equation}
    \label{eq:zcx-in-circ}
    \begin{tikzpicture}
	\begin{pgfonlayer}{nodelayer}
		\node [style=0 control] (0) at (0.25, 0.75) {\tiny$\neg$\textsf{0}};
		\node [style=smallbox] (1) at (0.25, -0.75) {$X$};
		\node [style=none] (10) at (2.45, 0.75) {};
		\node [style=none] (11) at (2.45, -0.75) {};
		\node [style=none] (18) at (-0.5, 0.75) {};
		\node [style=none] (19) at (1, 0.75) {};
		\node [style=none] (20) at (1, -0.75) {};
		\node [style=none] (21) at (-0.5, -0.75) {};
		\node [style=0 control] (26) at (-3.825, 0.75) {\tiny\textsf{0}};
		\node [style=smallbox] (27) at (-3.825, -0.75) {$X^{\dagger}$};
		\node [style=none] (28) at (-4.575, 0.75) {};
		\node [style=none] (29) at (-2.075, 0.75) {};
		\node [style=none] (30) at (-2.075, -0.75) {};
		\node [style=none] (31) at (-4.575, -0.75) {};
		\node [style=smallbox] (32) at (-2.775, -0.75) {$X$};
		\node [style=none] (33) at (-1.325, 0) {$=$};
	\end{pgfonlayer}
	\begin{pgfonlayer}{edgelayer}
		\draw [style=none] (1) to (0);
		\draw [style=none] (18.center) to (0);
		\draw [style=none] (0) to (19.center);
		\draw [style=none] (20.center) to (1);
		\draw [style=none] (1) to (21.center);
		\draw [style=none] (27) to (26);
		\draw [style=none] (28.center) to (26);
		\draw [style=none] (26) to (29.center);
		\draw [style=none] (30.center) to (27);
		\draw [style=none] (27) to (31.center);
	\end{pgfonlayer}
\end{tikzpicture}
  \end{equation}
  Therefore, multiplying the inverse of the circuit on the right-hand
  side of \cref{eq:zcx-circ-decomp-bor-pf} by an $X$ gate yields a
  representation of the $C_{\ket{0}}X$ over the gate $\s{X, CCX, H}$
  by \cref{lem:constructCX}. Hence, every matrix that can be
  represented by a circuit over $\s{X,C_{\ket{0}}X, H}$ can be
  represented by a circuit over $\s{X, CCX, H}$ using a single
  borrowed ancilla.
\end{proof}

\begin{remark}
  The construction in \cref{prop:construct-czerox} can be explained
  (and, in fact, was found) using the qupit ZH-calculus
  \cite{qupitZH}. In the qupit ZH-calculus, we have
  
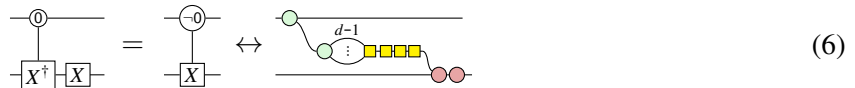
\begin{equation}
    \label{eq:zcx-in-zh}
    \begin{tikzpicture}
	\begin{pgfonlayer}{nodelayer}
		\node [style=none] (17) at (6.35, -0.125) {};
		\node [style=0 control] (0) at (0.25, 0.75) {\tiny$\neg$\textsf{0}};
		\node [style=smallbox] (1) at (0.25, -0.75) {$X$};
		\node [style=hadamard] (5) at (4.95, -0.125) {};
		\node [style=none] (10) at (2.45, 0.75) {};
		\node [style=none] (11) at (2.45, -0.75) {};
		\node [style=Z dot] (12) at (2.825, 0.75) {};
		\node [style=rn] (13) at (6.775, -0.75) {};
		\node [style=none] (14) at (7.4, 0.75) {};
		\node [style=none] (15) at (7.65, -0.75) {};
		\node [style=Z dot] (16) at (3.75, -0.125) {};
		\node [style=none] (18) at (-0.5, 0.75) {};
		\node [style=none] (19) at (1, 0.75) {};
		\node [style=none] (20) at (1, -0.75) {};
		\node [style=none] (21) at (-0.5, -0.75) {};
		\node [style=none] (22) at (1.75, 0) {$\leftrightarrow$};
		\node [style=none] (23) at (4.375, -0.125) {\tiny\rotatebox{90}{...}};
		\node [style=none] (24) at (4.35, 0.45) {\tiny${d\texttt{-}1}$};
		\node [style=0 control] (26) at (-3.825, 0.75) {\tiny\textsf{0}};
		\node [style=smallbox] (27) at (-3.825, -0.75) {$X^{\dagger}$};
		\node [style=none] (28) at (-4.575, 0.75) {};
		\node [style=none] (29) at (-2.075, 0.75) {};
		\node [style=none] (30) at (-2.075, -0.75) {};
		\node [style=none] (31) at (-4.575, -0.75) {};
		\node [style=smallbox] (32) at (-2.775, -0.75) {$X$};
		\node [style=none] (33) at (-1.325, 0) {$=$};
		\node [style=hadamard] (34) at (5.375, -0.125) {};
		\node [style=hadamard] (35) at (5.75, -0.125) {};
		\node [style=hadamard] (36) at (6.125, -0.125) {};
		\node [style=rn] (37) at (7.25, -0.75) {};
	\end{pgfonlayer}
	\begin{pgfonlayer}{edgelayer}
		\draw [style=none] (1) to (0);
		\draw [style=none] (10.center) to (12);
		\draw [style=none] (12) to (14.center);
		\draw [style=none] (15.center) to (13);
		\draw [style=none] (13) to (11.center);
		\draw [style=none, in=-180, out=-30, looseness=0.75] (12) to (16);
		\draw [style=none, in=150, out=-15, looseness=0.75] (17.center) to (13);
		\draw [bend left=45] (5) to (16);
		\draw [bend right=45] (5) to (16);
		\draw [style=none] (18.center) to (0);
		\draw [style=none] (0) to (19.center);
		\draw [style=none] (20.center) to (1);
		\draw [style=none] (1) to (21.center);
		\draw [style=none] (17.center) to (5);
		\draw [style=none] (27) to (26);
		\draw [style=none] (28.center) to (26);
		\draw [style=none] (26) to (29.center);
		\draw [style=none] (30.center) to (27);
		\draw [style=none] (27) to (31.center);
	\end{pgfonlayer}
\end{tikzpicture}
  \end{equation}
  where the $d-1$ label indicates there are $d-1$ number of wires in
  parallel. We then get the following construction of the $\ket{\neg
    0}$-controlled $X$ gate:
  
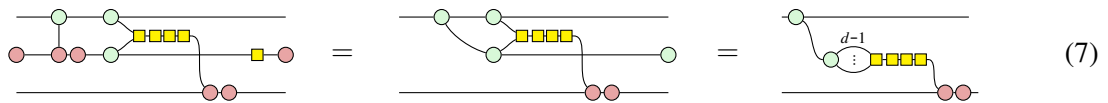
\begin{equation}
    \label{eq:postsel}
    \begin{tikzpicture}
	\begin{pgfonlayer}{nodelayer}
		\node [style=none] (11) at (-8.625, -1) {};
		\node [style=Z dot] (12) at (-6.125, 0) {};
		\node [style=rn] (13) at (-3.5, -1) {};
		\node [style=none] (15) at (-1.5, -1) {};
		\node [style=Z dot] (17) at (-6.125, 1) {};
		\node [style=none] (18) at (-1.5, 1) {};
		\node [style=hadamard] (19) at (-5.375, 0.5) {};
		\node [style=none] (20) at (-4, 0.5) {};
		\node [style=small hadamard] (21) at (-4.95, 0.5) {};
		\node [style=small hadamard] (22) at (-4.575, 0.5) {};
		\node [style=small hadamard] (23) at (-4.2, 0.5) {};
		\node [style=rn] (24) at (-3, -1) {};
		\node [style=none] (25) at (-8.625, 1) {};
		\node [style=Z dot] (26) at (-7.5, 1) {};
		\node [style=rn] (27) at (-7.5, 0) {};
		\node [style=rn] (28) at (-7, 0) {};
		\node [style=small hadamard] (29) at (-2.25, 0) {};
		\node [style=none] (30) at (0, 0) {$=$};
		\node [style=none] (32) at (1.5, -1) {};
		\node [style=Z dot] (33) at (4, 0) {};
		\node [style=rn] (34) at (6.625, -1) {};
		\node [style=Z dot] (35) at (8.625, 0) {};
		\node [style=none] (36) at (8.625, -1) {};
		\node [style=Z dot] (37) at (4, 1) {};
		\node [style=none] (38) at (8.625, 1) {};
		\node [style=hadamard] (39) at (4.75, 0.5) {};
		\node [style=none] (40) at (6.125, 0.5) {};
		\node [style=small hadamard] (41) at (5.175, 0.5) {};
		\node [style=small hadamard] (42) at (5.55, 0.5) {};
		\node [style=small hadamard] (43) at (5.925, 0.5) {};
		\node [style=rn] (44) at (7.125, -1) {};
		\node [style=none] (45) at (1.5, 1) {};
		\node [style=Z dot] (46) at (2.625, 1) {};
		\node [style=none] (47) at (15.525, -0.125) {};
		\node [style=hadamard] (48) at (14.125, -0.125) {};
		\node [style=none] (49) at (11.625, 1) {};
		\node [style=none] (50) at (11.625, -1) {};
		\node [style=Z dot] (51) at (12, 1) {};
		\node [style=rn] (52) at (15.95, -1) {};
		\node [style=none] (53) at (16.575, 1) {};
		\node [style=none] (54) at (16.825, -1) {};
		\node [style=Z dot] (55) at (12.925, -0.125) {};
		\node [style=none] (56) at (13.55, -0.125) {\tiny\rotatebox{90}{...}};
		\node [style=none] (57) at (13.525, 0.45) {\tiny${d\texttt{-}1}$};
		\node [style=hadamard] (58) at (14.55, -0.125) {};
		\node [style=hadamard] (59) at (14.925, -0.125) {};
		\node [style=hadamard] (60) at (15.3, -0.125) {};
		\node [style=rn] (61) at (16.425, -1) {};
		\node [style=none] (62) at (10.2, 0) {$=$};
		\node [style=rn] (63) at (-8.625, 0) {};
		\node [style=rn] (64) at (-1.5, 0) {};
	\end{pgfonlayer}
	\begin{pgfonlayer}{edgelayer}
		\draw [style=none] (15.center) to (13);
		\draw [style=none] (13) to (11.center);
		\draw [style=none] (17) to (18.center);
		\draw [style=none, in=150, out=-30, looseness=1.25] (17) to (19);
		\draw [style=none, in=-150, out=30, looseness=1.25] (12) to (19);
		\draw [style=none] (19) to (20.center);
		\draw [style=none, in=135, out=0, looseness=0.75] (20.center) to (13);
		\draw (25.center) to (26);
		\draw (26) to (17);
		\draw (12) to (28);
		\draw (28) to (27);
		\draw (27) to (26);
		\draw (29) to (12);
		\draw [style=none] (36.center) to (34);
		\draw [style=none] (34) to (32.center);
		\draw [style=none] (37) to (38.center);
		\draw [style=none, in=150, out=-30, looseness=1.25] (37) to (39);
		\draw [style=none, in=-150, out=30, looseness=1.25] (33) to (39);
		\draw [style=none] (39) to (40.center);
		\draw [style=none, in=135, out=0, looseness=0.75] (40.center) to (34);
		\draw (45.center) to (46);
		\draw (46) to (37);
		\draw [in=165, out=-60, looseness=0.75] (46) to (33);
		\draw (35) to (33);
		\draw [style=none] (49.center) to (51);
		\draw [style=none] (51) to (53.center);
		\draw [style=none] (54.center) to (52);
		\draw [style=none] (52) to (50.center);
		\draw [style=none, in=-180, out=-30, looseness=0.75] (51) to (55);
		\draw [style=none, in=150, out=-15, looseness=0.75] (47.center) to (52);
		\draw [bend left=45] (48) to (55);
		\draw [bend right=45] (48) to (55);
		\draw [style=none] (47.center) to (48);
		\draw (64) to (29);
		\draw (27) to (63);
	\end{pgfonlayer}
\end{tikzpicture}
  \end{equation}
  The post-selected circuit in \cref{eq:postsel} can be made
  deterministic by adding a $CX^\dagger$ gate for uncomputation, which
  yields a construction requiring a fresh ancilla:
  \begin{equation}
      \begin{tikzpicture}
	\begin{pgfonlayer}{nodelayer}
		\node [style=none] (11) at (-8.625, -1) {};
		\node [style=Z dot] (12) at (-6.125, 0) {};
		\node [style=rn] (13) at (-3.5, -1) {};
		\node [style=none] (15) at (-1.5, -1) {};
		\node [style=Z dot] (17) at (-6.125, 1) {};
		\node [style=none] (18) at (-1.5, 1) {};
		\node [style=hadamard] (19) at (-5.375, 0.5) {};
		\node [style=none] (20) at (-4, 0.5) {};
		\node [style=small hadamard] (21) at (-4.95, 0.5) {};
		\node [style=small hadamard] (22) at (-4.575, 0.5) {};
		\node [style=small hadamard] (23) at (-4.2, 0.5) {};
		\node [style=rn] (24) at (-3, -1) {};
		\node [style=none] (25) at (-8.625, 1) {};
		\node [style=Z dot] (26) at (-7.5, 1) {};
		\node [style=rn] (27) at (-7.5, 0) {};
		\node [style=rn] (28) at (-7, 0) {};
		\node [style=none] (30) at (0, 0) {$=$};
		\node [style=none] (32) at (1.5, -1) {};
		\node [style=Z dot] (33) at (4, 0) {};
		\node [style=rn] (34) at (6.625, -1) {};
		\node [style=Z dot] (35) at (7.625, 1) {};
		\node [style=none] (36) at (8.625, -1) {};
		\node [style=Z dot] (37) at (4, 1) {};
		\node [style=none] (38) at (8.625, 1) {};
		\node [style=hadamard] (39) at (4.75, 0.5) {};
		\node [style=none] (40) at (6.125, 0.5) {};
		\node [style=small hadamard] (41) at (5.175, 0.5) {};
		\node [style=small hadamard] (42) at (5.55, 0.5) {};
		\node [style=small hadamard] (43) at (5.925, 0.5) {};
		\node [style=rn] (44) at (7.125, -1) {};
		\node [style=none] (45) at (1.5, 1) {};
		\node [style=Z dot] (46) at (2.625, 1) {};
		\node [style=none] (47) at (15.525, -0.125) {};
		\node [style=hadamard] (48) at (14.125, -0.125) {};
		\node [style=none] (49) at (11.625, 1) {};
		\node [style=none] (50) at (11.625, -1) {};
		\node [style=Z dot] (51) at (12, 1) {};
		\node [style=rn] (52) at (15.95, -1) {};
		\node [style=none] (53) at (16.575, 1) {};
		\node [style=none] (54) at (16.825, -1) {};
		\node [style=Z dot] (55) at (12.925, -0.125) {};
		\node [style=none] (56) at (13.55, -0.125) {\tiny\rotatebox{90}{...}};
		\node [style=none] (57) at (13.525, 0.45) {\tiny${d\texttt{-}1}$};
		\node [style=hadamard] (58) at (14.55, -0.125) {};
		\node [style=hadamard] (59) at (14.925, -0.125) {};
		\node [style=hadamard] (60) at (15.3, -0.125) {};
		\node [style=rn] (61) at (16.425, -1) {};
		\node [style=none] (62) at (10.2, 0) {$=$};
		\node [style=rn] (63) at (-8.625, 0) {};
		\node [style=rn] (64) at (-1.5, 0) {};
		\node [style=Z dot] (65) at (-2.175, 1) {};
		\node [style=rn] (66) at (-2.175, 0) {};
		\node [style=rn] (67) at (-2.675, 0) {};
	\end{pgfonlayer}
	\begin{pgfonlayer}{edgelayer}
		\draw [style=none] (15.center) to (13);
		\draw [style=none] (13) to (11.center);
		\draw [style=none] (17) to (18.center);
		\draw [style=none, in=150, out=-30, looseness=1.25] (17) to (19);
		\draw [style=none, in=-150, out=30, looseness=1.25] (12) to (19);
		\draw [style=none] (19) to (20.center);
		\draw [style=none, in=135, out=0, looseness=0.75] (20.center) to (13);
		\draw (25.center) to (26);
		\draw (26) to (17);
		\draw (12) to (28);
		\draw (28) to (27);
		\draw (27) to (26);
		\draw [style=none] (36.center) to (34);
		\draw [style=none] (34) to (32.center);
		\draw [style=none] (37) to (38.center);
		\draw [style=none, in=150, out=-30, looseness=1.25] (37) to (39);
		\draw [style=none, in=-150, out=30, looseness=1.25] (33) to (39);
		\draw [style=none] (39) to (40.center);
		\draw [style=none, in=135, out=0, looseness=0.75] (40.center) to (34);
		\draw (45.center) to (46);
		\draw (46) to (37);
		\draw [in=165, out=-60, looseness=0.75] (46) to (33);
		\draw [in=0, out=-105, looseness=0.75] (35) to (33);
		\draw [style=none] (49.center) to (51);
		\draw [style=none] (51) to (53.center);
		\draw [style=none] (54.center) to (52);
		\draw [style=none] (52) to (50.center);
		\draw [style=none, in=-180, out=-30, looseness=0.75] (51) to (55);
		\draw [style=none, in=150, out=-15, looseness=0.75] (47.center) to (52);
		\draw [bend left=45] (48) to (55);
		\draw [bend right=45] (48) to (55);
		\draw [style=none] (47.center) to (48);
		\draw (27) to (63);
		\draw (67) to (66);
		\draw (66) to (65);
		\draw (64) to (66);
		\draw (67) to (12);
	\end{pgfonlayer}
\end{tikzpicture}
  \end{equation}
  The construction is then modified in order to work with a borrowed
  ancilla, which yields the circuit in
  \cref{eq:zcx-circ-decomp-bor-pf}.
\end{remark}    

By \cref{lem:constructCX,prop:construct-czerox}, the gate set
$\gatesetk{1}$ is equivalent (up to a borrowed ancilla) to the gate
set consisting of the $X$ gate, the $\ket{0}$-controlled $X$ gate, and
the Hadamard gate. Hence, by \cref{thm:daryrev}, any ternary classical
reversible function can be represented by a circuit over
$\gatesetk{1}$ using at most one borrowed ancilla.

We now show that, when $k\geq 2$, the Clifford-cyclotomic gate set of
degree $3^k$ is equivalent, up to a borrowed ancilla, to the
Clifford+$T_k$ gate set. We take advantage of some constructions from
\cite{BocharovA2017ternaryshor} (see, in particular, Figure~6 in
\cite{BocharovA2017ternaryshor}).

\begin{lemma}
  \label{lem:ZCX-CliffT}
  We have:
  \[
    \begin{tikzpicture}
	\begin{pgfonlayer}{nodelayer}
		\node [style=none] (1) at (-11.5, -0.625) {};
		\node [style=none] (2) at (-11.5, 0.625) {};
		\node [style=none] (4) at (-8.5, -0.625) {};
		\node [style=none] (5) at (-8.5, 0.625) {};
		\node [style=0 control] (6) at (-10, 0.625) {\textsf{0}};
		\node [style=box] (7) at (-10, -0.625) {$X$};
		\node [style=none] (12) at (-7, 0) {=};
		\node [style=none] (27) at (-5.75, -0.625) {};
		\node [style=none] (28) at (-5.75, 0.625) {};
		\node [style=box] (63) at (6.375, 0.625) {$X^2$};
		\node [style=box] (64) at (6.375, -0.625) {$H^3$};
		\node [style=none] (66) at (10.375, -0.625) {};
		\node [style=none] (67) at (10.375, 0.625) {};
		\node [style=box] (68) at (4.625, -0.625) {$T_2$};
		\node [style=0 control] (69) at (3.125, 0.625) {\tiny$\Lambda$};
		\node [style=box] (70) at (3.125, -0.625) {$X$};
		\node [style=box] (71) at (7.875, 0.625) {$S$};
		\node [style=box] (72) at (9.125, 0.625) {$X$};
		\node [style=box] (73) at (-4.375, -0.625) {$H$};
		\node [style=box] (74) at (1.625, -0.625) {$T_2$};
		\node [style=0 control] (75) at (0.125, 0.625) {\tiny$\Lambda$};
		\node [style=box] (76) at (0.125, -0.625) {$X$};
		\node [style=box] (77) at (-1.375, -0.625) {$T_2$};
		\node [style=0 control] (78) at (-2.875, 0.625) {\tiny$\Lambda$};
		\node [style=box] (79) at (-2.875, -0.625) {$X$};
	\end{pgfonlayer}
	\begin{pgfonlayer}{edgelayer}
		\draw (6) to (7);
		\draw (2.center) to (6);
		\draw (6) to (5.center);
		\draw (4.center) to (7);
		\draw (7) to (1.center);
		\draw (69) to (70);
		\draw (75) to (76);
		\draw (78) to (79);
		\draw (28.center) to (67.center);
		\draw (66.center) to (27.center);
	\end{pgfonlayer}
\end{tikzpicture}
  \]
\end{lemma}

\begin{lemma}
  \label{lem:sgate}
  We have:
  \[
    \begin{tikzpicture}
	\begin{pgfonlayer}{nodelayer}
		\node [style=none] (0) at (-11.5, 0.625) {};
		\node [style=none] (1) at (-11.5, -0.625) {};
		\node [style=none] (2) at (-8.5, 0.625) {};
		\node [style=none] (3) at (-8.5, -0.625) {};
		\node [style=box] (5) at (-10, 0.625) {$S$};
		\node [style=none] (6) at (-7, 0) {=};
		\node [style=none] (7) at (-5.75, -0.625) {};
		\node [style=none] (8) at (-5.75, 0.625) {};
		\node [style=none] (11) at (9.125, -0.625) {};
		\node [style=none] (12) at (9.125, 0.625) {};
		\node [style=0 control] (23) at (3.125, 0.625) {\tiny$\Lambda$};
		\node [style=box] (24) at (3.125, -0.625) {$X$};
		\node [style=box] (25) at (7.875, -0.625) {$X^2$};
		\node [style=box] (26) at (6.125, -0.625) {$H^2$};
		\node [style=box] (27) at (4.625, -0.625) {$X$};
		\node [style=box] (29) at (7.875, 0.625) {$X^2$};
		\node [style=box] (30) at (1.625, -0.625) {${T_2}^8$};
		\node [style=0 control] (31) at (0.125, 0.625) {\tiny$\Lambda$};
		\node [style=box] (32) at (0.125, -0.625) {$X$};
		\node [style=box] (33) at (-1.375, -0.625) {${T_2}^8$};
		\node [style=0 control] (34) at (-2.875, 0.625) {\tiny$\Lambda$};
		\node [style=box] (35) at (-2.875, -0.625) {$X$};
		\node [style=box] (36) at (-4.375, -0.625) {${T_2}^8$};
		\node [style=none] (42) at (2.1, -2.125) {{repeat $2$ times}};
		\node [style=none] (43) at (-5.325, -1.625) {};
		\node [style=none] (44) at (8.675, -1.625) {};
	\end{pgfonlayer}
	\begin{pgfonlayer}{edgelayer}
		\draw (2.center) to (5);
		\draw (5) to (0.center);
		\draw (23) to (24);
		\draw (8.center) to (12.center);
		\draw (11.center) to (7.center);
		\draw (1.center) to (3.center);
		\draw (31) to (32);
		\draw (34) to (35);
		\draw [style=brace edge] (44.center) to (43.center);
	\end{pgfonlayer}
\end{tikzpicture}
  \]
\end{lemma}

\begin{proposition}
  \label{prop:kgeq2equiv}
  When $k\geq 2$, the Clifford-cyclotomic gate set $\gatesetk{k}$ is
  equivalent to the Clifford+$T_k$ gate set up to a single borrowed
  ancilla.
\end{proposition}

\begin{proof}
  Recall that $\gatesetk{k}=\s{X, CX, CCX, H, \rotk{k}}$ and that
  Clifford+$\rotk{k}=\s{H, S, CX, \rotk{k}}$. To prove the
  proposition, we therefore need to show that the $S$ gate can be
  represented by a circuit over $\gatesetk{k}$ and that the $X$ and
  $CCX$ gates can be represented by Clifford+$\rotk{k}$ circuits. That
  the $S$ gate can be represented by a circuit over $\gatesetk{k}$
  follows from \cref{lem:sgate} and the fact that $T_2 =
  T_k^{3^{k-2}}$. That the $X$ can be represented by a
  Clifford+$\rotk{k}$ circuit simply follows from the fact that $X =
  H^\dagger T_2^3 H$. That the $CCX$ gate can be represented by a
  Clifford+$\rotk{k}$ circuit follows from \cref{lem:ZCX-CliffT} and
  \cref{thm:daryrev}.
\end{proof}

The propositions above show that there is some flexibility in the
definition of Clifford-cyclotomic gate sets and, in particular, that
the gate set $\s{X, CX, CCX, H, \rotk{k}}$ is by no means minimal.

\subsection{Circuit Representations for the Elements of \texorpdfstring{$\gens_{3^n}$}{G3n}}
\label{ssec:appcircs}

We now provide explicit constructions for the elements of
$\gens_{3^n}$. We focus on the matrices in $\gens_{3^n}$ where,
writing each computational basis state on $n$ qutrits as $n$ trits,
the levels are chosen to be those with the greatest value (taking the
last qutrit to have the least significant trit). Indeed, these
constructions can then be adapted to arbitrary levels by conjugating
them by ternary classical reversible circuits using
\cref{thm:daryrev,prop:construct-czerox}.

By \cref{thm:daryrev,prop:construct-czerox}, the multiply-controlled
$X$ gate can be expressed as a circuit over $\gatesetk{1}$ using a
single borrowed ancilla. We can therefore express the
multiply-controlled $Z$ gate as well, since $Z^\dagger = H X
H^\dagger$.

\begin{lemma}
  We have:
  
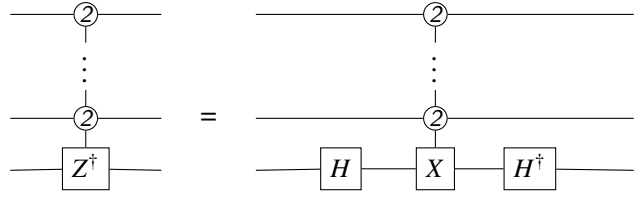
\begin{equation*}
      \begin{tikzpicture}
	\begin{pgfonlayer}{nodelayer}
		\node [style=none] (226) at (-9.25, 0) {};
		\node [style=none] (228) at (-5.25, 0) {};
		\node [style=0 control] (231) at (-7.25, 0) {\textsf{2}};
		\node [style=none] (232) at (-4, 0) {=};
		\node [style=none] (233) at (7.25, -1.375) {};
		\node [style=box] (234) at (-0.5, -1.375) {$H$};
		\node [style=box] (236) at (2, -1.375) {$X$};
		\node [style=box] (237) at (4.5, -1.375) {$H^{\dagger}$};
		\node [style=none] (241) at (-2.75, -1.375) {};
		\node [style=none] (243) at (-9.25, -1.375) {};
		\node [style=none] (244) at (-5.25, -1.375) {};
		\node [style=box] (258) at (-7.25, -1.375) {$Z^{\dagger}$};
		\node [style=none] (260) at (-9.25, 2.75) {};
		\node [style=none] (262) at (-5.25, 2.75) {};
		\node [style=0 control] (263) at (-7.25, 2.75) {\textsf{2}};
		\node [style=none] (268) at (-7.25, 2) {};
		\node [style=none] (269) at (-7.25, 0.75) {};
		\node [style=none] (270) at (-7.25, 1.5) {$\vdots$};
		\node [style=none] (271) at (-2.75, 0) {};
		\node [style=none] (272) at (7.25, 0) {};
		\node [style=0 control] (273) at (2, 0) {\textsf{2}};
		\node [style=none] (274) at (-2.75, 2.75) {};
		\node [style=none] (275) at (7.25, 2.75) {};
		\node [style=0 control] (276) at (2, 2.75) {\textsf{2}};
		\node [style=none] (277) at (2, 2) {};
		\node [style=none] (278) at (2, 0.75) {};
		\node [style=none] (279) at (2, 1.5) {$\vdots$};
	\end{pgfonlayer}
	\begin{pgfonlayer}{edgelayer}
		\draw (228.center) to (231);
		\draw (231) to (226.center);
		\draw [style=none] (241.center) to (234);
		\draw [style=none] (234) to (236);
		\draw [style=none] (236) to (237);
		\draw (244.center) to (258);
		\draw (258) to (231);
		\draw (258) to (243.center);
		\draw (260.center) to (263);
		\draw (263) to (262.center);
		\draw (233.center) to (237);
		\draw (263) to (268.center);
		\draw (269.center) to (231);
		\draw (272.center) to (273);
		\draw (273) to (271.center);
		\draw (274.center) to (276);
		\draw (276) to (275.center);
		\draw (276) to (277.center);
		\draw (278.center) to (273);
		\draw (273) to (236);
	\end{pgfonlayer}
\end{tikzpicture}
  \end{equation*}
\end{lemma}

From this and the fact that $(\omega)_{[2]}=S$ when acting on a single
qutrit, we can construct the 1-level matrix $(\omega)_{[x]}$ using a
single borrowed ancilla.

\begin{lemma}
  \label{lem:w1lvl}
  We have:
  
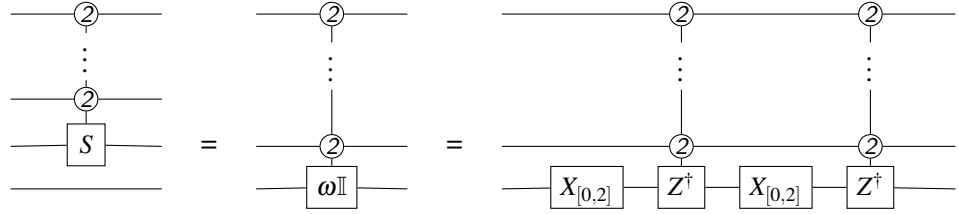
\begin{equation*}
      \begin{tikzpicture}
	\begin{pgfonlayer}{nodelayer}
		\node [style=none] (232) at (-4.5, -0.25) {=};
		\node [style=none] (233) at (8.75, -1.375) {};
		\node [style=box] (234) at (-1, -1.375) {$X_{[0,2]}$};
		\node [style=box] (236) at (1.5, -1.375) {$Z^{\dagger}$};
		\node [style=box] (237) at (4, -1.375) {$X_{[0,2]}$};
		\node [style=box] (239) at (6.5, -1.375) {$Z^{\dagger}$};
		\node [style=none] (241) at (-3.25, -1.375) {};
		\node [style=none] (243) at (-9.75, -1.375) {};
		\node [style=none] (244) at (-5.75, -1.375) {};
		\node [style=0 control] (245) at (1.5, -0.25) {\textsf{2}};
		\node [style=0 control] (246) at (6.5, -0.25) {\textsf{2}};
		\node [style=none] (247) at (8.75, -0.25) {};
		\node [style=none] (248) at (-3.25, -0.25) {};
		\node [style=none] (249) at (-16.25, -0.25) {};
		\node [style=none] (251) at (-12.25, -0.25) {};
		\node [style=box] (254) at (-14.25, -0.25) {$S$};
		\node [style=none] (255) at (-11, -0.25) {=};
		\node [style=none] (256) at (-16.25, -1.375) {};
		\node [style=none] (257) at (-12.25, -1.375) {};
		\node [style=box] (258) at (-7.75, -1.375) {$\omega \mathbb{I}$};
		\node [style=none] (259) at (-9.75, -0.25) {};
		\node [style=none] (260) at (-5.75, -0.25) {};
		\node [style=0 control] (261) at (-7.75, -0.25) {\textsf{2}};
		\node [style=none] (263) at (-9.75, 3.25) {};
		\node [style=none] (264) at (-5.75, 3.25) {};
		\node [style=0 control] (265) at (-7.75, 3.25) {\textsf{2}};
		\node [style=none] (266) at (-7.75, 2.5) {};
		\node [style=none] (267) at (-7.75, 1.25) {};
		\node [style=none] (268) at (-7.75, 2) {$\vdots$};
		\node [style=none] (281) at (-3.25, 3.25) {};
		\node [style=0 control] (283) at (1.5, 3.25) {\textsf{2}};
		\node [style=none] (284) at (1.5, 2.5) {};
		\node [style=none] (285) at (1.5, 1.25) {};
		\node [style=none] (286) at (1.5, 2) {$\vdots$};
		\node [style=0 control] (287) at (6.5, 3.25) {\textsf{2}};
		\node [style=none] (288) at (6.5, 2.5) {};
		\node [style=none] (289) at (6.5, 1.25) {};
		\node [style=none] (290) at (6.5, 2) {$\vdots$};
		\node [style=none] (291) at (8.75, 3.25) {};
		\node [style=none] (292) at (-16.25, 1) {};
		\node [style=none] (293) at (-12.25, 1) {};
		\node [style=0 control] (294) at (-14.25, 1) {\textsf{2}};
		\node [style=none] (295) at (-16.25, 3.25) {};
		\node [style=none] (296) at (-12.25, 3.25) {};
		\node [style=0 control] (297) at (-14.25, 3.25) {\textsf{2}};
		\node [style=none] (298) at (-14.25, 2.75) {};
		\node [style=none] (299) at (-14.25, 1.5) {};
		\node [style=none] (300) at (-14.25, 2.25) {$\vdots$};
	\end{pgfonlayer}
	\begin{pgfonlayer}{edgelayer}
		\draw [style=none] (241.center) to (234);
		\draw [style=none] (246) to (247.center);
		\draw [style=none] (245) to (236);
		\draw [style=none] (246) to (239);
		\draw [style=none] (234) to (236);
		\draw [style=none] (236) to (237);
		\draw [style=none] (239) to (237);
		\draw [style=none] (239) to (233.center);
		\draw [style=none] (248.center) to (245);
		\draw [style=none] (245) to (246);
		\draw (251.center) to (254);
		\draw (254) to (249.center);
		\draw [style=none] (256.center) to (257.center);
		\draw (244.center) to (258);
		\draw (258) to (243.center);
		\draw (260.center) to (261);
		\draw (261) to (259.center);
		\draw (263.center) to (265);
		\draw (265) to (264.center);
		\draw (265) to (266.center);
		\draw (267.center) to (261);
		\draw (281.center) to (283);
		\draw (283) to (284.center);
		\draw (287) to (288.center);
		\draw (291.center) to (287);
		\draw (287) to (283);
		\draw (285.center) to (245);
		\draw (289.center) to (246);
		\draw (261) to (258);
		\draw (293.center) to (294);
		\draw (294) to (292.center);
		\draw (295.center) to (297);
		\draw (297) to (296.center);
		\draw (297) to (298.center);
		\draw (299.center) to (294);
		\draw (294) to (254);
	\end{pgfonlayer}
\end{tikzpicture}
  \end{equation*}
\end{lemma}

The next two lemmas let us construct the 1-level matrix
$(-1)_{[x]}$. When acting on a single qutrit, this is the $(-1)_{[2]}
= \text{diag}(1,1,-1)$ gate. This gate is also known as the
\textbf{metaplectic} gate
\cite{BocharovA2016optimalitymetaplectic,BocharovA2017ternaryshor,CuiS2015universalmetaplectic}
and in earlier work, we referred to this gate as the $R$ gate
\cite{metaplectic}.

\begin{lemma}
  \label{lem:TCHmw}
  We have:
  \begin{equation*}
      \scalebox{0.95}{\begin{tikzpicture}
	\begin{pgfonlayer}{nodelayer}
		\node [style=none] (0) at (-4.5, 2.375) {};
		\node [style=none] (1) at (-4.5, -1.25) {};
		\node [style=none] (2) at (-4.5, 0) {};
		\node [style=none] (3) at (-1.5, 2.375) {};
		\node [style=none] (4) at (-1.5, -1.25) {};
		\node [style=none] (5) at (-1.5, 0) {};
		\node [style=0 control] (6) at (-3, 0) {\textsf{2}};
		\node [style=box] (7) at (-3, -1.25) {$-\omega H$};
		\node [style=none] (8) at (-3, 1.75) {};
		\node [style=none] (9) at (-3, 0.625) {};
		\node [style=none] (10) at (-3, 1.375) {$\vdots$};
		\node [style=0 control] (11) at (-3, 2.375) {\textsf{2}};
		\node [style=none] (12) at (0, 0) {=};
		\node [style=box] (169) at (2.625, -1.25) {${H^\dagger}^2$};
		\node [style=0 control] (174) at (4.375, 0) {\textsf{2}};
		\node [style=box] (175) at (4.375, -1.25) {$S^2$};
		\node [style=none] (176) at (4.375, 1.75) {};
		\node [style=none] (177) at (4.375, 0.625) {};
		\node [style=none] (178) at (4.375, 1.375) {$\vdots$};
		\node [style=0 control] (179) at (4.375, 2.375) {\textsf{2}};
		\node [style=box] (181) at (6.125, -1.25) {${H^\dagger}^2$};
		\node [style=0 control] (186) at (7.875, 0) {\textsf{2}};
		\node [style=box] (187) at (7.875, -1.25) {$S^2$};
		\node [style=none] (188) at (7.875, 1.75) {};
		\node [style=none] (189) at (7.875, 0.625) {};
		\node [style=none] (190) at (7.875, 1.375) {$\vdots$};
		\node [style=0 control] (191) at (7.875, 2.375) {\textsf{2}};
		\node [style=none] (205) at (1.25, -1.25) {};
		\node [style=none] (206) at (1.25, 0) {};
		\node [style=none] (207) at (1.25, 2.375) {};
		\node [style=box] (208) at (11.625, -1.25) {${H^\dagger}^2$};
		\node [style=0 control] (209) at (13.375, 0) {\textsf{2}};
		\node [style=box] (210) at (13.375, -1.25) {$S^2$};
		\node [style=none] (211) at (13.375, 1.75) {};
		\node [style=none] (212) at (13.375, 0.625) {};
		\node [style=none] (213) at (13.375, 1.375) {$\vdots$};
		\node [style=0 control] (214) at (13.375, 2.375) {\textsf{2}};
		\node [style=box] (215) at (15.125, -1.25) {${H^\dagger}^2$};
		\node [style=0 control] (216) at (16.875, 0) {\textsf{2}};
		\node [style=box] (217) at (16.875, -1.25) {$S^2$};
		\node [style=none] (218) at (16.875, 1.75) {};
		\node [style=none] (219) at (16.875, 0.625) {};
		\node [style=none] (220) at (16.875, 1.375) {$\vdots$};
		\node [style=0 control] (221) at (16.875, 2.375) {\textsf{2}};
		\node [style=box] (222) at (9.625, -1.25) {$H$};
		\node [style=box] (223) at (18.875, -1.25) {$H^\dagger$};
		\node [style=box] (224) at (21.125, -1.25) {${H^\dagger}^2$};
		\node [style=0 control] (225) at (22.875, 0) {\textsf{2}};
		\node [style=box] (226) at (22.875, -1.25) {$S^2$};
		\node [style=none] (227) at (22.875, 1.75) {};
		\node [style=none] (228) at (22.875, 0.625) {};
		\node [style=none] (229) at (22.875, 1.375) {$\vdots$};
		\node [style=0 control] (230) at (22.875, 2.375) {\textsf{2}};
		\node [style=box] (231) at (24.625, -1.25) {${H^\dagger}^2$};
		\node [style=0 control] (232) at (26.375, 0) {\textsf{2}};
		\node [style=box] (233) at (26.375, -1.25) {$S^2$};
		\node [style=none] (234) at (26.375, 1.75) {};
		\node [style=none] (235) at (26.375, 0.625) {};
		\node [style=none] (236) at (26.375, 1.375) {$\vdots$};
		\node [style=0 control] (237) at (26.375, 2.375) {\textsf{2}};
		\node [style=none] (238) at (27.625, 2.375) {};
		\node [style=none] (239) at (27.625, -1.25) {};
		\node [style=none] (240) at (27.625, 0) {};
	\end{pgfonlayer}
	\begin{pgfonlayer}{edgelayer}
		\draw (6) to (7);
		\draw (2.center) to (6);
		\draw (6) to (5.center);
		\draw (4.center) to (7);
		\draw (7) to (1.center);
		\draw (9.center) to (6);
		\draw (0.center) to (11);
		\draw (11) to (3.center);
		\draw (8.center) to (11);
		\draw (177.center) to (174);
		\draw (176.center) to (179);
		\draw (189.center) to (186);
		\draw (188.center) to (191);
		\draw [style=none] (169) to (175);
		\draw [style=none] (175) to (181);
		\draw [style=none] (187) to (181);
		\draw [style=none] (179) to (191);
		\draw [style=none] (174) to (186);
		\draw (186) to (187);
		\draw (174) to (175);
		\draw (179) to (207.center);
		\draw (206.center) to (174);
		\draw (169) to (205.center);
		\draw (212.center) to (209);
		\draw (211.center) to (214);
		\draw (219.center) to (216);
		\draw (218.center) to (221);
		\draw [style=none] (208) to (210);
		\draw [style=none] (210) to (215);
		\draw [style=none] (217) to (215);
		\draw [style=none] (214) to (221);
		\draw [style=none] (209) to (216);
		\draw (216) to (217);
		\draw (209) to (210);
		\draw (228.center) to (225);
		\draw (227.center) to (230);
		\draw (235.center) to (232);
		\draw (234.center) to (237);
		\draw [style=none] (224) to (226);
		\draw [style=none] (226) to (231);
		\draw [style=none] (233) to (231);
		\draw [style=none] (230) to (237);
		\draw [style=none] (225) to (232);
		\draw (232) to (233);
		\draw (225) to (226);
		\draw (191) to (214);
		\draw (186) to (209);
		\draw (187) to (222);
		\draw (222) to (208);
		\draw (221) to (230);
		\draw (216) to (225);
		\draw (217) to (223);
		\draw (223) to (224);
		\draw (239.center) to (233);
		\draw (232) to (240.center);
		\draw (238.center) to (237);
	\end{pgfonlayer}
\end{tikzpicture}}
  \end{equation*}
\end{lemma}

\begin{lemma}
  \label{lem:minus1}
  We have:
  \begin{equation*}
      \begin{tikzpicture}
	\begin{pgfonlayer}{nodelayer}
		\node [style=0 control] (12) at (-2.75, -0.5) {\textsf{2}};
		\node [style=box] (13) at (-2.75, -2) {$-\mathbb{I}$};
		\node [style=none] (14) at (-4.25, -2) {};
		\node [style=none] (15) at (-4.25, -0.5) {};
		\node [style=none] (16) at (-1.25, -0.5) {};
		\node [style=none] (17) at (-1.25, -2) {};
		\node [style=none] (18) at (0, 0) {$=$};
		\node [style=none] (25) at (-5.5, 0) {$=$};
		\node [style=box] (26) at (-8.75, -0.5) {$(-1)_{[2]}$};
		\node [style=none] (27) at (-10.5, -2) {};
		\node [style=none] (28) at (-10.5, -0.5) {};
		\node [style=none] (29) at (-7, -0.5) {};
		\node [style=none] (30) at (-7, -2) {};
		\node [style=0 control] (31) at (3.25, -0.5) {\textsf{2}};
		\node [style=box] (32) at (3.25, -2) {$-\omega H$};
		\node [style=none] (33) at (1.25, -2) {};
		\node [style=none] (34) at (1.25, -0.5) {};
		\node [style=0 control] (35) at (9, -0.5) {\textsf{2}};
		\node [style=box] (36) at (9, -2) {$X_{[1,2]}$};
		\node [style=none] (37) at (10.625, -0.5) {};
		\node [style=none] (38) at (10.625, -2) {};
		\node [style=0 control] (40) at (6.25, -0.5) {\textsf{2}};
		\node [style=box] (41) at (6.25, -2) {$-\omega H$};
		\node [style=0 control] (42) at (3.25, 2.75) {\textsf{2}};
		\node [style=none] (43) at (3.25, 2) {};
		\node [style=none] (44) at (3.25, 0.75) {};
		\node [style=none] (45) at (3.25, 1.5) {$\vdots$};
		\node [style=0 control] (46) at (6.25, 2.75) {\textsf{2}};
		\node [style=none] (47) at (6.25, 2) {};
		\node [style=none] (48) at (6.25, 0.75) {};
		\node [style=none] (49) at (6.25, 1.5) {$\vdots$};
		\node [style=0 control] (50) at (9, 2.75) {\textsf{2}};
		\node [style=none] (51) at (9, 2) {};
		\node [style=none] (52) at (9, 0.75) {};
		\node [style=none] (53) at (9, 1.5) {$\vdots$};
		\node [style=0 control] (54) at (-2.75, 2.75) {\textsf{2}};
		\node [style=none] (55) at (-2.75, 2) {};
		\node [style=none] (57) at (-2.75, 1.5) {$\vdots$};
		\node [style=none] (61) at (-10.5, 0.5) {};
		\node [style=none] (62) at (-7, 0.5) {};
		\node [style=0 control] (63) at (-8.75, 0.5) {\textsf{2}};
		\node [style=none] (64) at (-10.5, 2.75) {};
		\node [style=none] (65) at (-7, 2.75) {};
		\node [style=0 control] (66) at (-8.75, 2.75) {\textsf{2}};
		\node [style=none] (67) at (-8.75, 2.25) {};
		\node [style=none] (68) at (-8.75, 1) {};
		\node [style=none] (69) at (-8.75, 1.75) {$\vdots$};
		\node [style=none] (71) at (-1.25, 2.75) {};
		\node [style=none] (72) at (-4.25, 2.75) {};
		\node [style=none] (73) at (-2.75, 0.75) {};
		\node [style=none] (74) at (1.25, 2.75) {};
		\node [style=none] (75) at (10.625, 2.75) {};
	\end{pgfonlayer}
	\begin{pgfonlayer}{edgelayer}
		\draw (12) to (13);
		\draw (14.center) to (13);
		\draw [style=none] (15.center) to (12);
		\draw (13) to (17.center);
		\draw (16.center) to (12);
		\draw [style=none] (28.center) to (26);
		\draw (29.center) to (26);
		\draw (27.center) to (30.center);
		\draw (31) to (32);
		\draw (33.center) to (32);
		\draw [style=none] (34.center) to (31);
		\draw (35) to (36);
		\draw (36) to (38.center);
		\draw (37.center) to (35);
		\draw (40) to (41);
		\draw (31) to (40);
		\draw (40) to (35);
		\draw (32) to (41);
		\draw (41) to (36);
		\draw (42) to (43.center);
		\draw (46) to (47.center);
		\draw (50) to (51.center);
		\draw (54) to (55.center);
		\draw (62.center) to (63);
		\draw (63) to (61.center);
		\draw (64.center) to (66);
		\draw (66) to (65.center);
		\draw (66) to (67.center);
		\draw (68.center) to (63);
		\draw (63) to (26);
		\draw (71.center) to (54);
		\draw (54) to (72.center);
		\draw (73.center) to (12);
		\draw (44.center) to (31);
		\draw (48.center) to (40);
		\draw (52.center) to (35);
		\draw (50) to (46);
		\draw (46) to (42);
		\draw (42) to (74.center);
		\draw (75.center) to (50);
	\end{pgfonlayer}
\end{tikzpicture}
  \end{equation*}
\end{lemma}

We can now synthesize the 3-level matrix $H_{[x_1,x_2,x_3]}$ matrix
over $\gatesetk{1}$. To do this, apply \cref{lem:TCHmw} as well as the
appropriate controlled global phase correction: a product of 1-level
$\omega_{[x]}$ matrices and $(-1)_{[x]}$ matrices.

\begin{lemma}
  \label{lem:TCH}
  We have:
  
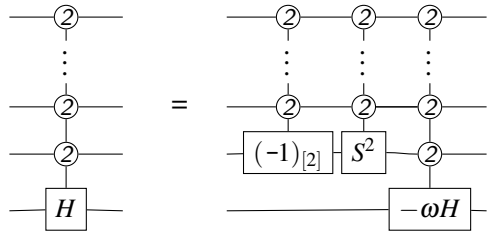
\begin{equation*}
      \begin{tikzpicture}
	\begin{pgfonlayer}{nodelayer}
		\node [style=none] (0) at (-5, 1.375) {};
		\node [style=none] (1) at (-5, -3.75) {};
		\node [style=none] (2) at (-5, -1) {};
		\node [style=none] (3) at (-2, 1.375) {};
		\node [style=none] (4) at (-2, -3.75) {};
		\node [style=none] (5) at (-2, -1) {};
		\node [style=0 control] (6) at (-3.5, -1) {\textsf{2}};
		\node [style=box] (7) at (-3.5, -3.75) {$H$};
		\node [style=none] (8) at (-3.5, 0.75) {};
		\node [style=none] (9) at (-3.5, -0.375) {};
		\node [style=none] (10) at (-3.5, 0.375) {$\vdots$};
		\node [style=0 control] (11) at (-3.5, 1.375) {\textsf{2}};
		\node [style=none] (12) at (-0.5, -1) {=};
		\node [style=0 control] (174) at (2.375, -1) {\textsf{2}};
		\node [style=box] (175) at (2.375, -2.25) {$(\texttt{-}1)_{[2]}$};
		\node [style=none] (176) at (2.375, 0.75) {};
		\node [style=none] (177) at (2.375, -0.375) {};
		\node [style=none] (178) at (2.375, 0.375) {$\vdots$};
		\node [style=0 control] (179) at (2.375, 1.375) {\textsf{2}};
		\node [style=0 control] (186) at (4.375, -1) {\textsf{2}};
		\node [style=box] (187) at (4.375, -2.25) {$S^2$};
		\node [style=none] (188) at (4.375, 0.75) {};
		\node [style=none] (189) at (4.375, -0.375) {};
		\node [style=none] (190) at (4.375, 0.375) {$\vdots$};
		\node [style=0 control] (191) at (4.375, 1.375) {\textsf{2}};
		\node [style=none] (205) at (0.75, -2.25) {};
		\node [style=none] (206) at (0.75, -1) {};
		\node [style=none] (207) at (0.75, 1.375) {};
		\node [style=none] (243) at (-5, -2.25) {};
		\node [style=none] (244) at (-2, -2.25) {};
		\node [style=0 control] (245) at (-3.5, -2.25) {\textsf{2}};
		\node [style=none] (246) at (7.625, 1.375) {};
		\node [style=none] (247) at (7.625, -3.75) {};
		\node [style=none] (248) at (7.625, -1) {};
		\node [style=0 control] (249) at (6.125, -1) {\textsf{2}};
		\node [style=box] (250) at (6.125, -3.75) {$-\omega H$};
		\node [style=none] (251) at (6.125, 0.75) {};
		\node [style=none] (252) at (6.125, -0.375) {};
		\node [style=none] (253) at (6.125, 0.375) {$\vdots$};
		\node [style=0 control] (254) at (6.125, 1.375) {\textsf{2}};
		\node [style=none] (255) at (7.625, -2.25) {};
		\node [style=0 control] (256) at (6.125, -2.25) {\textsf{2}};
		\node [style=none] (257) at (0.75, -3.75) {};
	\end{pgfonlayer}
	\begin{pgfonlayer}{edgelayer}
		\draw (2.center) to (6);
		\draw (6) to (5.center);
		\draw (4.center) to (7);
		\draw (7) to (1.center);
		\draw (9.center) to (6);
		\draw (0.center) to (11);
		\draw (11) to (3.center);
		\draw (8.center) to (11);
		\draw (177.center) to (174);
		\draw (176.center) to (179);
		\draw (189.center) to (186);
		\draw (188.center) to (191);
		\draw [style=none] (179) to (191);
		\draw [style=none] (174) to (186);
		\draw (186) to (187);
		\draw (174) to (175);
		\draw (179) to (207.center);
		\draw (206.center) to (174);
		\draw (243.center) to (245);
		\draw (245) to (244.center);
		\draw (6) to (245);
		\draw (245) to (7);
		\draw (249) to (248.center);
		\draw (247.center) to (250);
		\draw (252.center) to (249);
		\draw (254) to (246.center);
		\draw (251.center) to (254);
		\draw (256) to (255.center);
		\draw (249) to (256);
		\draw (256) to (250);
		\draw (257.center) to (250);
		\draw [in=360, out=180] (256) to (187);
		\draw (187) to (175);
		\draw (175) to (205.center);
		\draw (186) to (249);
		\draw (191) to (254);
		\draw (186) to (249);
	\end{pgfonlayer}
\end{tikzpicture}
  \end{equation*}
\end{lemma}

We have now constructed all of the required 1-, 2-, and 3-level
matrices (up to a permutation). We can therefore prove
\cref{prop:circuits}, which we restate below, making the ancilla
requirements explicit.

\begin{proposition*}
  If $U\in\gensnn{3^n}$, then $U$ can be represented by a circuit over
  $\gatesetk{1}$ using at most 2 borrowed ancillae. Explicitly,
  \begin{itemize}
      \item $(-1)_{[x]}$ requires 2 borrowed ancillae,  
      \item $(\omega)_{[x]}$ requires 1 borrowed ancillae,
      \item $X_{[x_1,x_2]}$ requires 1 borrowed ancilla, and
      \item $H_{[x_1,x_2,x_3]}$ requires 1 borrowed ancillae.
  \end{itemize}
\end{proposition*}

\begin{proof}
  This follows from
  \cref{lem:constructCX,prop:construct-czerox,thm:daryrev,lem:w1lvl,lem:minus1,lem:TCH}.
\end{proof}

The number of ancillae required to represent the elements of
$\gensnn{3^n}$ is, to a certain extent, an artifact of the choice of
gate set. For example, including the $\ket{0}$-controlled $X$ gate to
the gate set would lower the ancilla-count for some of the elements of
$\gensnn{3^n}$.

The proposition above shows that the matrices that can be represented
by a multiqutrit circuit over the Clifford+$(-1)_{[2]}$ gate set (also
known as the \textbf{Clifford+$R$} or the \textbf{metaplectic} gate
set) are a subset of those representable by a circuit over
$\gatesetk{1}$. At the time of writing, we do not know whether this
inclusion is strict, although the conjecture in
\cite{BocharovA2016ternaryarithmetics} that not all ternary classical
reversible gates can be exactly represented over the
Clifford+$(-1)_{[2]}$ gate set lends credence to this idea.

If a matrix can be represented by a circuit over $\gatesetk{k}$, it
can also be represented by a circuit over $\gatesetk{k+1}$. It
therefore follows from the proposition above that all of the elements
of $\gensnn{3^n}$ can be represented by a circuit over
$\gatesetk{2}$. We close this appendix by showing that the 1-level
matrix $(\omega_2)_{[x]}$ can be represented by a circuit over
$\gatesetk{2}$ and by providing further generalizations of the above
constructions. This paves the way for a direct proof of exact
synthesis for Clifford+$T$ circuits (rather than the more indirect one
using catalytic embeddings, as in \cref{thm:characterization}). Over
$\gatesetk{2}$, the ancilla requirements are lowered, since the
$\ket{0}$-controlled $X$ gate can be represented by an ancilla-free
circuit by \cref{lem:ZCX-CliffT}. To construct $(\omega_2)_{[x]}$, we
first build a modification of $(\omega)_{[x]}$ which differs by a
controlled global phase of $\omega_2$.

\begin{lemma}
  \label{lem:w2Sdag}
  We have:
  
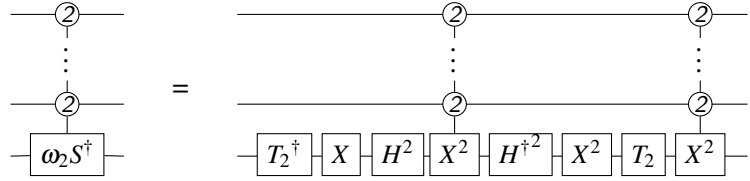
\begin{equation*}
      \begin{tikzpicture}
	\begin{pgfonlayer}{nodelayer}
		\node [style=none] (0) at (-11.5, -1.75) {};
		\node [style=none] (2) at (-8.5, -1.75) {};
		\node [style=box] (5) at (-10, -1.75) {$\omega_2 S^\dagger$};
		\node [style=none] (6) at (-7, 0) {=};
		\node [style=none] (8) at (-5.5, -1.75) {};
		\node [style=none] (12) at (8, -1.75) {};
		\node [style=box] (29) at (-4.25, -1.75) {${T_2}^\dagger$};
		\node [style=box] (30) at (-2.75, -1.75) {$X$};
		\node [style=box] (31) at (-1.25, -1.75) {$H^2$};
		\node [style=box] (32) at (0.25, -1.75) {$X^2$};
		\node [style=box] (33) at (2, -1.75) {${H^\dagger}^2$};
		\node [style=box] (34) at (3.75, -1.75) {$X^2$};
		\node [style=box] (35) at (5.25, -1.75) {$T_2$};
		\node [style=box] (36) at (6.75, -1.75) {$X^2$};
		\node [style=0 control] (39) at (6.75, -0.375) {\textsf{2}};
		\node [style=0 control] (56) at (6.75, 2) {\textsf{2}};
		\node [style=none] (63) at (6.75, 1.375) {};
		\node [style=none] (64) at (6.75, 0.25) {};
		\node [style=none] (65) at (6.75, 1) {$\vdots$};
		\node [style=none] (72) at (8, -0.375) {};
		\node [style=0 control] (73) at (0.25, -0.375) {\textsf{2}};
		\node [style=none] (74) at (-5.5, -0.375) {};
		\node [style=none] (75) at (8, 2) {};
		\node [style=0 control] (76) at (0.25, 2) {\textsf{2}};
		\node [style=none] (77) at (-5.5, 2) {};
		\node [style=none] (78) at (0.25, 1.375) {};
		\node [style=none] (79) at (0.25, 0.25) {};
		\node [style=none] (80) at (0.25, 1) {$\vdots$};
		\node [style=none] (81) at (-8.5, -0.375) {};
		\node [style=0 control] (82) at (-10, -0.375) {\textsf{2}};
		\node [style=none] (83) at (-11.5, -0.375) {};
		\node [style=none] (84) at (-8.5, 2) {};
		\node [style=0 control] (85) at (-10, 2) {\textsf{2}};
		\node [style=none] (86) at (-11.5, 2) {};
		\node [style=none] (87) at (-10, 1.375) {};
		\node [style=none] (88) at (-10, 0.25) {};
		\node [style=none] (89) at (-10, 1) {$\vdots$};
	\end{pgfonlayer}
	\begin{pgfonlayer}{edgelayer}
		\draw (2.center) to (5);
		\draw (5) to (0.center);
		\draw (8.center) to (12.center);
		\draw (56) to (63.center);
		\draw (64.center) to (39);
		\draw (39) to (36);
		\draw [in=180, out=0] (74.center) to (72.center);
		\draw (77.center) to (75.center);
		\draw (76) to (78.center);
		\draw (79.center) to (73);
		\draw (73) to (32);
		\draw [in=180, out=0] (83.center) to (81.center);
		\draw (86.center) to (84.center);
		\draw (85) to (87.center);
		\draw (88.center) to (82);
		\draw [in=90, out=-90] (82) to (5);
	\end{pgfonlayer}
\end{tikzpicture}
  \end{equation*}
\end{lemma}

We note that unlike the construction in Lemma~\ref{lem:w1lvl} which
required one (additional) borrowed ancilla, this construction requires
no (additional) borrowed ancillae. By combining the construction of
\cref{lem:w2Sdag} and that of \cref{lem:TCH}, we can therefore
represent $H_{[x_1,x_2,x_3]}$ without ancillae. Similarly, by
combining the construction of \cref{lem:w2Sdag} and that of
\cref{lem:minus1}, we can represent $(-1)_{[x]}$ using a single
borrowed ancillae. Finally, $(\omega_2)_{[x]}$ can be constructed as
in the next lemma using 2 borrowed ancillae.

\begin{lemma}
  \label{lem:w2}
  We have:
  
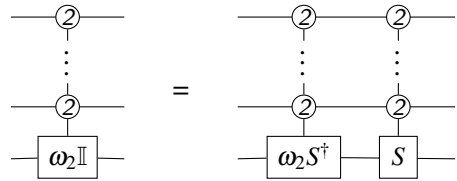
\begin{equation*}
      \begin{tikzpicture}
	\begin{pgfonlayer}{nodelayer}
		\node [style=none] (0) at (-6.5, -1.75) {};
		\node [style=none] (2) at (-3.5, -1.75) {};
		\node [style=box] (5) at (-5, -1.75) {$\omega_2 \mathbb{I}$};
		\node [style=none] (6) at (-2, 0) {=};
		\node [style=none] (81) at (-3.5, -0.375) {};
		\node [style=0 control] (82) at (-5, -0.375) {\textsf{2}};
		\node [style=none] (83) at (-6.5, -0.375) {};
		\node [style=none] (84) at (-3.5, 2) {};
		\node [style=0 control] (85) at (-5, 2) {\textsf{2}};
		\node [style=none] (86) at (-6.5, 2) {};
		\node [style=none] (87) at (-5, 1.375) {};
		\node [style=none] (88) at (-5, 0.25) {};
		\node [style=none] (89) at (-5, 1) {$\vdots$};
		\node [style=none] (93) at (-0.5, -1.75) {};
		\node [style=none] (94) at (5.25, -1.75) {};
		\node [style=box] (95) at (1.25, -1.75) {$\omega_2 S^\dagger$};
		\node [style=none] (96) at (5.25, -0.375) {};
		\node [style=0 control] (97) at (1.25, -0.375) {\textsf{2}};
		\node [style=none] (98) at (-0.5, -0.375) {};
		\node [style=none] (99) at (5.25, 2) {};
		\node [style=0 control] (100) at (1.25, 2) {\textsf{2}};
		\node [style=none] (101) at (-0.5, 2) {};
		\node [style=none] (102) at (1.25, 1.375) {};
		\node [style=none] (103) at (1.25, 0.25) {};
		\node [style=none] (104) at (1.25, 1) {$\vdots$};
		\node [style=box] (108) at (3.75, -1.75) {$S$};
		\node [style=0 control] (109) at (3.75, -0.375) {\textsf{2}};
		\node [style=0 control] (110) at (3.75, 2) {\textsf{2}};
		\node [style=none] (111) at (3.75, 1.375) {};
		\node [style=none] (112) at (3.75, 0.25) {};
		\node [style=none] (113) at (3.75, 1) {$\vdots$};
	\end{pgfonlayer}
	\begin{pgfonlayer}{edgelayer}
		\draw (2.center) to (5);
		\draw (5) to (0.center);
		\draw [in=180, out=0] (83.center) to (81.center);
		\draw (86.center) to (84.center);
		\draw (85) to (87.center);
		\draw (88.center) to (82);
		\draw (94.center) to (95);
		\draw (95) to (93.center);
		\draw [in=180, out=0] (98.center) to (96.center);
		\draw (101.center) to (99.center);
		\draw (100) to (102.center);
		\draw (103.center) to (97);
		\draw (110) to (111.center);
		\draw (112.center) to (109);
		\draw (82) to (5);
		\draw (97) to (95);
		\draw (108) to (109);
	\end{pgfonlayer}
\end{tikzpicture}
  \end{equation*}
\end{lemma}

\begin{proposition}
  \label{prop:app-ct}
  The 1-, 2-, and 3-level matrices $(-1)_{[x]}$, $(\omega_2)_{[x]}$,
  $X_{[x_1, x_2]}$, and $H_{[x_1,x_2,x_3]}$ can be represented by a
  circuit over $\gatesetk{2}$ using at most 2 borrowed
  ancillae. Explicitly,
  \begin{itemize}
      \item $(-1)_{[x]}$ requires 1 borrowed ancilla, 
      \item $(\omega_2)_{[x]}$ requires 2 borrowed ancillae, 
      \item $X_{[x_1,x_2]}$ requires 0 borrowed ancillae, and
      \item $H_{[x_1,x_2,x_3]}$ requires 0 borrowed ancillae.
  \end{itemize}
\end{proposition}

\begin{proof}
  This follows from
  \cref{lem:ZCX-CliffT,lem:w2Sdag,lem:w2,lem:minus1,lem:TCH,thm:daryrev}.
\end{proof}

We can generalize the above construction to Clifford-cyclotomic gate
sets of higher degree.

\begin{proposition}
  Let $k\geq 1$. The 1-level matrix $(\omega_k)_{[x]}$ can be
  represented by a circuit over $\gatesetk{k}$ using $k$ borrowed
  ancillae.
\end{proposition}

\begin{proof}
  First, we build the multiply-controlled $M$ gate, where $M =
  \textnormal{diag}(1,\omega_k,\omega_k^\dagger)$.
  \begin{equation}
    \begin{tikzpicture}
	\begin{pgfonlayer}{nodelayer}
		\node [style=box] (13) at (-2.75, -2) {$T_k^\dagger$};
		\node [style=none] (14) at (-4.25, -2) {};
		\node [style=none] (15) at (-4.25, -0.5) {};
		\node [style=none] (25) at (-5.5, 0) {$=$};
		\node [style=none] (27) at (-10.5, -2) {};
		\node [style=none] (28) at (-10.5, -0.5) {};
		\node [style=none] (29) at (-7, -0.5) {};
		\node [style=none] (30) at (-7, -2) {};
		\node [style=0 control] (31) at (0, -0.5) {\textsf{2}};
		\node [style=box] (32) at (0, -2) {$H^2$};
		\node [style=0 control] (35) at (5.75, -0.5) {\textsf{2}};
		\node [style=box] (36) at (5.75, -2) {${H^\dagger}^2$};
		\node [style=none] (37) at (7.375, -0.5) {};
		\node [style=none] (38) at (7.375, -2) {};
		\node [style=box] (41) at (3, -2) {$T_k$};
		\node [style=0 control] (42) at (0, 2.75) {\textsf{2}};
		\node [style=none] (43) at (0, 2) {};
		\node [style=none] (44) at (0, 0.75) {};
		\node [style=none] (45) at (0, 1.5) {$\vdots$};
		\node [style=0 control] (50) at (5.75, 2.75) {\textsf{2}};
		\node [style=none] (51) at (5.75, 2) {};
		\node [style=none] (52) at (5.75, 0.75) {};
		\node [style=none] (53) at (5.75, 1.5) {$\vdots$};
		\node [style=none] (64) at (-10.5, 2.75) {};
		\node [style=none] (65) at (-7, 2.75) {};
		\node [style=none] (72) at (-4.25, 2.75) {};
		\node [style=none] (75) at (7.375, 2.75) {};
		\node [style=0 control] (76) at (-8.625, -0.5) {\textsf{2}};
		\node [style=0 control] (77) at (-8.625, 2.75) {\textsf{2}};
		\node [style=none] (78) at (-8.625, 2) {};
		\node [style=none] (79) at (-8.625, 0.75) {};
		\node [style=none] (80) at (-8.625, 1.5) {$\vdots$};
		\node [style=box] (81) at (-8.625, -2) {$M$};
	\end{pgfonlayer}
	\begin{pgfonlayer}{edgelayer}
		\draw (14.center) to (13);
		\draw (31) to (32);
		\draw (35) to (36);
		\draw (36) to (38.center);
		\draw (37.center) to (35);
		\draw (32) to (41);
		\draw (41) to (36);
		\draw (42) to (43.center);
		\draw (50) to (51.center);
		\draw (44.center) to (31);
		\draw (52.center) to (35);
		\draw (75.center) to (50);
		\draw (72.center) to (42);
		\draw (42) to (50);
		\draw (35) to (31);
		\draw (31) to (15.center);
		\draw (13) to (32);
		\draw (77) to (78.center);
		\draw (79.center) to (76);
		\draw (64.center) to (77);
		\draw (77) to (65.center);
		\draw (29.center) to (76);
		\draw (76) to (28.center);
		\draw (27.center) to (81);
		\draw (81) to (30.center);
		\draw (76) to (81);
	\end{pgfonlayer}
\end{tikzpicture}
  \end{equation}
  Then, we can build the multiply-controlled one-qutrit gate
  $\omega_k(\omega_{k-1})^\dagger_{[2]} = \omega_k
  \textnormal{diag}(1,1,\omega_{k-1}^\dagger)$.
  \begin{equation}
    \begin{tikzpicture}
	\begin{pgfonlayer}{nodelayer}
		\node [style=box] (0) at (-2.25, -1.5) {$X$};
		\node [style=none] (1) at (-3.75, -1.5) {};
		\node [style=none] (2) at (-3.75, 0) {};
		\node [style=none] (3) at (-5, 0.5) {$=$};
		\node [style=none] (4) at (-10, -1.5) {};
		\node [style=none] (5) at (-10, 0) {};
		\node [style=none] (6) at (-6.5, 0) {};
		\node [style=none] (7) at (-6.5, -1.5) {};
		\node [style=0 control] (8) at (2.75, 0) {\textsf{2}};
		\node [style=box] (9) at (2.75, -1.5) {$M$};
		\node [style=0 control] (10) at (10, 0) {\textsf{2}};
		\node [style=box] (11) at (10, -1.5) {$M$};
		\node [style=none] (12) at (11.625, 0) {};
		\node [style=none] (13) at (11.625, -1.5) {};
		\node [style=0 control] (15) at (2.75, 3.25) {\textsf{2}};
		\node [style=none] (16) at (2.75, 2.5) {};
		\node [style=none] (17) at (2.75, 1.25) {};
		\node [style=none] (18) at (2.75, 2) {$\vdots$};
		\node [style=0 control] (19) at (10, 3.25) {\textsf{2}};
		\node [style=none] (20) at (10, 2.5) {};
		\node [style=none] (21) at (10, 1.25) {};
		\node [style=none] (22) at (10, 2) {$\vdots$};
		\node [style=none] (23) at (-10, 3.25) {};
		\node [style=none] (24) at (-6.5, 3.25) {};
		\node [style=none] (25) at (-3.75, 3.25) {};
		\node [style=none] (26) at (11.625, 3.25) {};
		\node [style=0 control] (27) at (-8.125, 0) {\textsf{2}};
		\node [style=0 control] (28) at (-8.125, 3.25) {\textsf{2}};
		\node [style=none] (29) at (-8.125, 2.5) {};
		\node [style=none] (30) at (-8.125, 1.25) {};
		\node [style=none] (31) at (-8.125, 2) {$\vdots$};
		\node [style=box] (32) at (-8.125, -1.5) {$\omega_k (\omega_{k-1})_{[2]}^\dagger$};
		\node [style=box] (33) at (-0.875, -1.5) {$H^2$};
		\node [style=box] (34) at (0.625, -1.5) {$X^\dagger$};
		\node [style=box] (35) at (4.625, -1.5) {$X$};
		\node [style=box] (36) at (6.25, -1.5) {${H^\dagger}^2$};
		\node [style=box] (37) at (8, -1.5) {$X^\dagger$};
	\end{pgfonlayer}
	\begin{pgfonlayer}{edgelayer}
		\draw (1.center) to (0);
		\draw (8) to (9);
		\draw (10) to (11);
		\draw (11) to (13.center);
		\draw (12.center) to (10);
		\draw (15) to (16.center);
		\draw (19) to (20.center);
		\draw (17.center) to (8);
		\draw (21.center) to (10);
		\draw (26.center) to (19);
		\draw (25.center) to (15);
		\draw (15) to (19);
		\draw (10) to (8);
		\draw (8) to (2.center);
		\draw (0) to (9);
		\draw (28) to (29.center);
		\draw (30.center) to (27);
		\draw (23.center) to (28);
		\draw (28) to (24.center);
		\draw (6.center) to (27);
		\draw (27) to (5.center);
		\draw (4.center) to (32);
		\draw (32) to (7.center);
		\draw (27) to (32);
		\draw (9) to (35);
		\draw (35) to (36);
		\draw (36) to (37);
		\draw (37) to (11);
	\end{pgfonlayer}
\end{tikzpicture}
  \end{equation}
  Finally, we can combine this with the multiply-controlled one-qutrit
  gate $(\omega_{k-1})_{[2]}=\textnormal{diag}(1,1,\omega_{k-1})$ to get
  $(\omega_k)_{[2...2]}$.
  \begin{equation}
    \begin{tikzpicture}
	\begin{pgfonlayer}{nodelayer}
		\node [style=0 control] (12) at (-2.75, -0.5) {\textsf{2}};
		\node [style=box] (13) at (-2.75, -2) {$\omega_k \mathbb{I}$};
		\node [style=none] (14) at (-4.25, -2) {};
		\node [style=none] (15) at (-4.25, -0.5) {};
		\node [style=none] (16) at (-1.25, -0.5) {};
		\node [style=none] (17) at (-1.25, -2) {};
		\node [style=none] (18) at (0, 0) {$=$};
		\node [style=none] (25) at (-5.5, 0) {$=$};
		\node [style=box] (26) at (-8.75, -0.5) {$(\omega_k)_{[2]}$};
		\node [style=none] (27) at (-10.5, -2) {};
		\node [style=none] (28) at (-10.5, -0.5) {};
		\node [style=none] (29) at (-7, -0.5) {};
		\node [style=none] (30) at (-7, -2) {};
		\node [style=0 control] (31) at (3.75, -0.5) {\textsf{2}};
		\node [style=box] (32) at (3.75, -2) {$\omega_k (\omega_{k-1})_{[2]}^\dagger$};
		\node [style=none] (33) at (1.25, -2) {};
		\node [style=none] (34) at (1.25, -0.5) {};
		\node [style=none] (37) at (11, -0.5) {};
		\node [style=none] (38) at (11, -2) {};
		\node [style=0 control] (40) at (8.5, -0.5) {\textsf{2}};
		\node [style=box] (41) at (8.5, -2) {$(\omega_{k-1})_{[2]}$};
		\node [style=0 control] (42) at (3.75, 2.75) {\textsf{2}};
		\node [style=none] (43) at (3.75, 2) {};
		\node [style=none] (44) at (3.75, 0.75) {};
		\node [style=none] (45) at (3.75, 1.5) {$\vdots$};
		\node [style=0 control] (46) at (8.5, 2.75) {\textsf{2}};
		\node [style=none] (47) at (8.5, 2) {};
		\node [style=none] (48) at (8.5, 0.75) {};
		\node [style=none] (49) at (8.5, 1.5) {$\vdots$};
		\node [style=0 control] (54) at (-2.75, 2.75) {\textsf{2}};
		\node [style=none] (55) at (-2.75, 2) {};
		\node [style=none] (57) at (-2.75, 1.5) {$\vdots$};
		\node [style=none] (61) at (-10.5, 0.5) {};
		\node [style=none] (62) at (-7, 0.5) {};
		\node [style=0 control] (63) at (-8.75, 0.5) {\textsf{2}};
		\node [style=none] (64) at (-10.5, 2.75) {};
		\node [style=none] (65) at (-7, 2.75) {};
		\node [style=0 control] (66) at (-8.75, 2.75) {\textsf{2}};
		\node [style=none] (67) at (-8.75, 2.25) {};
		\node [style=none] (68) at (-8.75, 1) {};
		\node [style=none] (69) at (-8.75, 1.75) {$\vdots$};
		\node [style=none] (71) at (-1.25, 2.75) {};
		\node [style=none] (72) at (-4.25, 2.75) {};
		\node [style=none] (73) at (-2.75, 0.75) {};
		\node [style=none] (74) at (1.25, 2.75) {};
		\node [style=none] (75) at (11, 2.75) {};
	\end{pgfonlayer}
	\begin{pgfonlayer}{edgelayer}
		\draw (12) to (13);
		\draw (14.center) to (13);
		\draw [style=none] (15.center) to (12);
		\draw (13) to (17.center);
		\draw (16.center) to (12);
		\draw [style=none] (28.center) to (26);
		\draw (29.center) to (26);
		\draw (27.center) to (30.center);
		\draw (31) to (32);
		\draw (33.center) to (32);
		\draw [style=none] (34.center) to (31);
		\draw (40) to (41);
		\draw (31) to (40);
		\draw (32) to (41);
		\draw (42) to (43.center);
		\draw (46) to (47.center);
		\draw (54) to (55.center);
		\draw (62.center) to (63);
		\draw (63) to (61.center);
		\draw (64.center) to (66);
		\draw (66) to (65.center);
		\draw (66) to (67.center);
		\draw (68.center) to (63);
		\draw (63) to (26);
		\draw (71.center) to (54);
		\draw (54) to (72.center);
		\draw (73.center) to (12);
		\draw (44.center) to (31);
		\draw (48.center) to (40);
		\draw (46) to (42);
		\draw (42) to (74.center);
		\draw (75.center) to (46);
		\draw (40) to (37.center);
		\draw (38.center) to (41);
	\end{pgfonlayer}
\end{tikzpicture}
  \end{equation}
  Since a single borrowed ancilla suffices to build $(\omega)$ and 2
  borrowed ancillae suffice to build $(\omega_2)$, the above equation
  shows that $k$ ancillae suffice to build $(\omega_k)$.
\end{proof}

\end{document}